\documentclass[11pt]{amsart}
\usepackage[a4paper]{geometry}              % See geometry.pdf to learn the layout options. There are lots.       % ... or a4paper or a5paper or ... 
\usepackage{graphicx}
\usepackage{natbib}
\usepackage{amssymb}
\usepackage{amsaddr} %authors affiliations
\usepackage{amsmath,mdframed}
\usepackage{amsfonts}
\usepackage{color}
\usepackage{xcolor}
\usepackage{bbm}
\usepackage{setspace}
\usepackage{appendix}
\usepackage{hyperref}  
\usepackage{tikz}  
\usetikzlibrary{%
	calc,%
	decorations.pathmorphing,%
		decorations.pathreplacing,%
	fadings,%
	shadings,%
	arrows,shapes,positioning
	}
\usepackage{circuitikz}
\usepackage{subfig}
\usepackage[normalem]{ulem}
\usepackage{datetime}
\usepackage[utf8]{inputenc}
\usepackage{fancyhdr}
\usepackage[retainorgcmds]{IEEEtrantools}
\usepackage{apptools}
\newcounter{shmanman}
\newtheorem{theorem}{Theorem}%[section]

\newtheorem{claim}[shmanman]{Claim}

\newtheorem{corollary}[theorem]{Corollary}

\newtheorem{definition}{Definition}

\newtheorem{lemma}[theorem]{Lemma}

\newtheorem{proposition}[theorem]{Proposition}

\hypersetup{
    colorlinks=true,
    linkcolor=red,
    citecolor=blue,
    filecolor=magenta,      
    urlcolor=blue,
}

\usepackage{enumitem}
\setlist[itemize]{label=---}

\allowdisplaybreaks

\title[]{Active Product Development\\
\vspace{0.5cm}
\small{July 2025}
\vspace{-2cm}
}

\author{Santiago Oliveros$^{\dagger}$}
\thanks{$^{\dagger}$ Department of Economics, University of Bristol, United Kingdom, e-mail: \href{mailto: s.oliveros@brisatol.ac.uk}{ s.oliveros@bristol.ac.uk}.\\
 I thank Christian Ewerhart, Hans Gersbach, Alastair Langtry, Barton Lee, Marek Pycia, Mauricio Ribeiro, and audiences at University of Bristol, and ETH Zurich.}

\begin{document}

\maketitle

\begin{abstract}

We introduce a dynamic model in which a developer incrementally improves a product of uncertain quality over time, with the quality evolving as a controlled Brownian motion. At each moment in time, the developer can continue exploring by paying a flow cost, restart from a previously attained quality level by paying a fixed cost, or terminate the process by either freely abandoning the project or by incurring a cost to launch the highest quality observed so far. The optimal strategy is characterized by a free boundary of an impulse-controlled Brownian motion, reflecting how the developer’s tolerance for setbacks evolves over time in three distinct stages. In the early stage, setbacks lead to project termination; in the intermediate stage, unsuccessful paths serve as learning opportunities, leading to strategic restarts; and in the final stage, the product is inevitably launched, with further improvements enhancing its final quality in a final push. Our model captures the role of active decision-making in managing uncertainty through a true process of trial and error that can be reversed when necessary. The results highlight the essential interplay of persistence and chance, demonstrating that success hinges not only on avoiding prolonged failure but also on the precise timing of interventions.

\end{abstract}

\medskip

\hangindent=2cm{
 Innovation, R\&D.}

\vspace{0.5cm}

\hangindent=2cm{ \textit{JEL:} \hspace{0.85cm} C72, C73; D81, D82; O31, O32. }

\thispagestyle{empty}

\pagebreak

%\section{}
%\subsection{}

\pagenumbering{arabic}

\bigskip

\section{Introduction}

\medskip

\medskip

Innovation is frequently portrayed as the outcome of collaboration within firms, research institutions, or organized R\&D teams. Yet a substantial share originates from individuals acting independently—such as inventors, garage tinkerers, or user-developers—who face uncertain market rewards and are motivated by curiosity or a desire to solve practical problems. These individuals, who operate outside conventional institutional frameworks, have a significant and measurable impact: \cite{vonhippel2005,vonhippel2007book,frazierhippel2018} document widespread user innovation, while W.I.P.O [\citeyear{wipo2019}] shows that a consistent share of U.S. patents are filed by individuals. Iconic examples include the Apple I, developed by Jobs and Wozniak in a garage, which helped launch an entire industry.

Even within firms, innovation is not always a top-down process. Employees may pursue personal or side projects—a phenomenon referred to as “bootlegging” (\cite{criscuoloetal2014})—which operate independently of formal incentives or structured R\&D systems. \cite{amabile2018} emphasizes that creative output often stems from intrinsic rather than extrinsic motivation, with employees frequently solving problems that go beyond their formal job descriptions. A classic case is 3M’s Post-it Note, which emerged from a failed adhesive experiment.\footnote{Dr. Spencer Silver discovered a low-tack adhesive in 1968 and, after informally promoting the idea, his colleague Art Fry recognized its potential as a bookmark in 1974, leading to the commercial launch of the \textit{Post-it Note} in 1980. \cite{postitwiki}} Such bottom-up innovation underscores the need to integrate non-traditional pathways into economic models of technological change.

Innovative efforts—whether by firms or individuals—typically evolve through iterative trial and error. As \cite{thomke1998,lochetal2023} argue, developers must experiment, assess outcomes, and sometimes revert back to prior versions if there are no improvements. In this process, failures provide feedback that guide future refinements.\footnote{Thomas Edison’s testing of filament materials exemplifies this: “I have not failed. I've just found 10,000 ways that won't work.”} In short, innovation involves learning from setbacks, discarding failed attempts, and building upon partial successes to converge on viable solutions. By embracing uncertainty and learning from both success and failure, innovation becomes a strategic process of adaptation and persistence.\footnote{\cite{zuckerberg2017speech}: “Facebook wasn’t the first thing I built. I also built games, chat systems, study tools and music players… I’m not alone. J.K. Rowling got rejected 12 times before publishing Harry Potter. Even Beyoncé had to make hundreds of songs to get ‘Halo.’”} 

The selective persistence of the innovation process reveals a psychological-economic trade-off between \textit{grit} and the capacity to \textit{quit}. developers must choose whether to persevere through uncertainty or abandon unworkable ideas. While persistence may yield breakthroughs, it also risks sunk-cost traps and emotional overinvestment (\cite{duckworth2016}). Navigating this tension is essential for understanding the behavior of both solo and organizational developers, and for designing institutions that support exploration and timely exit.

In this paper, we examine the decision-making process of an intrinsically motivated developer who develops and has the opportunity to launch a product herself.\footnote{Although the developer develops and markets the product, we abstract from standard marketing considerations. See \cite{Dahanhauser2002} on \textit{end-to-end} innovation in the managerial literature.} The development process requires active decision-making to assess whether to continue along the current trajectory, revert to a previously successful version, or abandon the project. Continued development incurs a flow cost, regardless of progress, while launching or restarting the project entails fixed costs, prompting the developer to intervene only intermittently.

We model product development as a Brownian motion with known drift and variance, where the available versions of the product are given by the managed path. Each version's quality maps to a known profit function. Launching involves selecting the best realized version; development creates new versions. At any point, the developer faces a dynamic decision: (i) allowing for the stochastic evolution of the development process, (ii) stopping the process and either launch the current best version or abandon the project for free, or (iii) restarting the development process from a prior best version at a fixed adjustment cost.

Given the positive relationship between quality and profitability, the developer will always launch of restart the process from the highest quality attained. Thus, intervention is governed by a free boundary, a threshold that depends on both the current and maximum attained quality. When the quality drifts too far below the maximum, intervention becomes optimal. If the developer decides to abandon the process, she gets zero payoff and sunk costs are not recovered. If she decides to launch, the process ends and she collects the profits of the launched product net of sunk costs. 

The developer's problem is then an impulse control problem (\cite{harrisonetal1983})—not a passive observation—where the developer intervenes strategically to maximize expected payoff. Since the achieved maximum determines the outside option's value, the free boundary is continuous and only piecewise differentiable—reflecting intervention points. The model captures the developer’s active, intermittent control over product evolution, emphasizing the balance between perseverance and disengagement, and the tension between grit and the desire to quit.

We identify three distinct stages in the development path, each characterized by evolving tolerance for setbacks.\footnote{\cite{book:grubermarquis1969}, \cite{utterback1971}, and \cite{kavadiasH-K2020} identify three phases in innovation: idea generation, problem solving, and implementation.} In the early stage, the free boundary is decreasing in the maximum value; once the process hits it, the project is aborted. Here, major setbacks lead to termination—\textit{quitting}. In the intermediate stage, the developer does not launch or quit but iterates. Restarting from the highest quality prior version is optimal. This stage is characterized by iterative experimentation where setbacks serve as learning opportunities with a developer that shows optimal grit. Over time, tolerance for setbacks decreases and restarts become more frequent. In this stage, the free boundary is an increasing function of the highest attained quality. In the final stage, the product is viable: continued progress enhances its quality, while persistent setbacks prompt a launch rather than a restart of the development process. The free boundary is still increasing in the highest attained quality leading to a higher likelihood of launch over time. This reflects Steve Jobs’ mantra of narrowing focus.\footnote{Steve Jobs: ``Be a yardstick of quality. Some people aren't used to an environment where excellence is expected.''}

Our model highlights the dual role of tenacity and luck in innovation, and the importance of active participation in determining when to persist, restart, or quit. The developer is lucky in avoiding early setbacks and in restarting optimally when quality is high. Since the process is forward-looking, optimal intervention timing is only clear ex post. In hindsight, the developer can assess decisions, appreciating the value of persistence or recognizing that quitting earlier would have been better due to high costs incurred.

\medskip

The rest of the paper is organized as follows. In the next section, we position our work within the broader academic literature. Section~\ref{sec: model} introduces the model and the underlying stochastic processes, establishing the theoretical framework that supports our analysis. We then proceed to solve the model in a stepwise fashion. First, we derive the necessary and sufficient conditions for a candidate value function to represent the optimal strategy. We then characterize this function in terms of a free boundary, which gives rise to a differential equation governing its evolution. Although a closed-form solution is not available, we provide a qualitative analysis of the developer's behavior implied by the structure of the free boundary. Section~\ref{sec: compa} examines how changes in the fixed costs of restarting development and launching the product influence the developer’s optimal decisions. All technical proofs are provided in the Appendix.

\section{Literature Review}

The unpredictable nature of research and development (R\&D) has been central to understanding irreversible investments in innovation. The real options literature (\cite{dixitpindyck1994}) highlights the importance of managerial flexibility in decisions to delay, expand, or abandon innovation projects in response to changing market conditions. Complementing this perspective, the experimentation literature (\cite{kelleretal2005, boltonharris1999}) emphasizes the uncertain and iterative aspects of innovation, particularly within cooperative organizational settings. This literature underscores the role of decision-making and team dynamics in refining projects through continuous information gathering and feedback. Our model aligns with these two strands by incorporating a stochastic process with observable outcomes to explicitly capture the effects of uncertainty on innovation.

\medskip

Unlike the experimentation literature, which primarily focuses on the mean and variance of a stochastic process, the type of uncertainty relevant for payoffs in our model is given by particular outcomes along the realized path. In this sense, our model connects with an expanding body of work initiated by \cite{callander2011a}, which conceptualizes innovation as a trial-and-error process. In \cite{callander2011a}, short-lived agents arrive sequentially, each sampling one product and providing information to subsequent agents.\footnote{\cite{garfagninistrulovivi2016} also study learning with path dependency, but instead of single-period agents, they assume agents live for two periods, increasing their incentive for experimentation relative to exploitation. \cite{eratkavadias2008} also examine sequential testing with correlated tests, but under a distinct framework for experimentation.} Through this sequential sampling, agents incrementally improve their understanding of how product features relate to payoffs. This analytical framework has been applied in diverse contexts, including policy learning (\cite{callanderhummel2014, callander2011b}), expert decision-making under complex communication structures (\cite{callanderetal2021}), and strategic organizational learning (\cite{callandermatouschek2019}). 

\medskip

Other papers in this tradition examine learning about stochastic processes, but rather than focusing on particular realizations, agents' payoffs depend on the entire realized sample path. For instance, \cite{callanderclark2017} extend \cite{callander2011a} to study the evolution of legal doctrine through sequences of resolved cases; \cite{bardhibobkova2023} explore how policymakers interpret correlated evidence from citizens about local conditions; \cite{bardhi2024} analyzes multi-attribute product design; \cite{ilutvanchev2023} consider a related design problem from a collective decision-making perspective; and \cite{carnehlschneider2021} analyze knowledge creation through pairs of questions and answers.\footnote{This research also relates to literature investigating how market conditions influence the direction of innovation (\cite{hopenhaynsquintani2021, hopenhaynetal2006, bryanlemus2017}). Our paper shifts the focus from external market conditions to internal agent-driven processes.}

\medskip

In our paper, the decision-maker similarly engages in trial-and-error learning. However, rather than sampling from a predetermined set of unknown product types, the developer can restart the development process upon encountering setbacks. Specifically, if the current path experiences a sharp decline, the decision-maker identifies can restart from a previously successful point. This introduces a dynamic selection mechanism where past successes serve as anchor points for future iterations, reflecting a path actively shaped rather than inherited. Moreover, it is the developer who sets the actual correlation between past and present successes.

\medskip

In both the real options and experimentation literature, payoffs remain stochastic even after adopting an innovation. In real options, the final value of the stochastic process matters, whereas in experimentation, only the learned drift influences long-term payoffs. In both cases, the specific realized path is irrelevant for future outcomes. By contrast, in the trial-and-error literature, the realized path itself is crucial since only previously observed outcomes inform future choices. Our model similarly makes future payoffs dependent on the realized path, but notably, the decisive factor is the highest realized value.

\medskip

A recent set of papers examine similar environments where the highest achieved value primarily determines future payoffs. In \cite{urcunyariv2023} and \cite{wong2024}, search occurs with recall, allowing termination at any time while preserving the highest observed value. They focus on the ``speed'' of innovation by a single agent which proxies the spread of the potential draws by affecting the variance of a drift-less brownian motion. \cite{cetemenetal2023} and \cite{li2021} study related problems but focus on multiple agents collectively \textit{riding} the sample path.

Similar to these studies, our decision-maker prioritizes the highest achieved value, with behavior influenced by the gap between current and peak values (drawdown process), correlated values and perfect recall.\footnote{\cite{fengetal2024} analyze incentives for agents required to achieve an exogenous milestone, with setbacks affecting progress. Unlike this approach, our termination decision is endogenous and driven by natural developments rather than incentive structures.} However, unlike these studies, our decision-maker actively intervenes, restarting the process when encountering setbacks. Thus, rather than passively recalling values, the developer strategically shapes their path. Technically, the underlying drawdown process in our model differs from prior work\footnote{See for example, \cite{urcunyariv2023} and \cite{cetemenetal2023}.} by employing impulse-controlled Brownian motion rather than following a fixed path. This distinction highlights the active and adaptive nature of decision-making within our framework, distinguishing it from traditional recall-based models.\footnote{In models following \cite{callander2011a}, the developer observes past trials and solves a filtering problem instead of a controlled problem. On a technical note, our paper employs a construction similar to \cite{book:harrison2013}, but with a dynamically evolving free boundary.}

Substantively, our paper contributes to the literature on multistage product development by endogenizing the stages of development. \cite{koussisetal2013} study improvements in already marketed products modeled as a pure jump process with four exogenous phases, including two (\textit{exploration} and \textit{value enhancing}) analogous to the final two stages of our model.\footnote{See also \cite{herathpark2002} for a broader taxonomy of product development phases.} Consequently, although our paper does not deal with optimal contracts, it also informs research on incentives within multistage projects:\footnote{See, for example, \cite{greentaylor2016, curellosinander2020, moroni2022, madsen2022}.} when quality choices are endogenous, product development may naturally require multiple \textit{endogenous} distinct stages with crucial intermediate milestones.

\section{The model} \label{sec: model}

An agent (which we will refer to as developer or innovator) is responsible for developing and introducing a product to the market. We define $\mathcal{Q}_{t}$ as the set of available product qualities or versions at time $t$, whose dynamics are governed by a controlled Brownian motion, formally defined below. The development process occurs in continuous time and incurs a continuous flow cost $c>0$ as long as the developer does not actively intervene in the process.

An active intervention entails halting the ongoing development process and making a strategic decision. The developer has the flexibility to restart development from any previously attained quality level within the set $\mathcal{Q}_{t}$. Alternatively, the developer may terminate the development process and introduce the product to the market at a quality level within $\mathcal{Q}_{t}$ or abandon the project entirely. Restarting the development process is an instantaneous action that incurs a fixed cost $R > c$. Similarly, launching the product is an instantaneous decision that entails a one-time cost $L > R$. Lastly, the developer always retains the costless option to discontinue the development process without bringing the product to market.

A product of quality $q_t$ delivers a flow payoffs of $\pi(q_t)$ for eternity. We assume that $\pi(\cdot)$ is $C^2$ and strictly concave and that it achieves a maximum at some finite $\overline{q}>0$. Hence, $\pi(\cdot)$ is bounded above. The developer discounts payoffs at the rate $r>0$ and we normalize flow payoffs by the factor $r$. For the problem to be meaningful, we assume that $\pi(\overline{q}) >> L$.

Given the properties of the profit function and that any meaningful intervention (either to launch the product or to restart the development process) entails paying a fixed cost, the only relevant quality for decision making is the highest available quality: $q^*_t = \sup \left\{  \mathcal{Q}_t \right\}$.

\subsection{Development process.} The set of available versions (qualities) of the product at time $t \geq 0$ is given by
\[
\mathcal{Q}_{t\geq 0} \equiv \left\{ y : \exists s \leq t, y = Y_s \right\}
\]
where $Y_t$ is a L\'evy process of the form
\[
Y_t = X_t + U_t
\]
where $X_t$ is a $(\mu,\sigma)$ brownian motion and $U_t$ is a pure jump process, both to be defined below.  Let $W_t$ be the standard Wiener process adapted to a filtered space $(\Omega,\mathcal{F},\mathbb{F},\mathbb{P})$ for the filtration $\mathbb{F} = \left\{ \mathcal{F}_t \right\}_{t \geq 0}$ and let $X_t$ be described by the stochastic differential equation
\[
dX_t = \mu dt + \sigma dWt
\]
To simplify the analysis, we assume throughout the paper that $\mu>0$. The only results that require reinterpretation are the comparative statics analysis while the rest of the paper does not depend on this assumption. We model the interventions $U_t$ as an adapted (to $\mathbb{F} $) pure jump process that changes discontinuously at the moment of intervention:\footnote{The pure jump nature of the process $U_t$ follows by the assumption that each intervention incurs a fixed cost.} $U_t$ is a collection of stopping times $\mathcal{T}$ and a jump at each stopping time $\tau \in \mathcal{T}$ with size determined below.

For any arbitrary collection of stopping times $\mathcal{T}$, it is useful to define two auxiliary processes, $Z_t$ and $M_t$:
\begin{enumerate}
\item $Z_t \leq 0 \leq M_t$ for any $t \geq 0$
\item $Z_\tau \equiv 0$ for any $\tau \in \mathcal{T}$, and
\begin{align*}
dZ_t \equiv &
\begin{cases}
dX_t & \text{ if } Z_t <0\\
\min \left\{ 0, dX_t \right\} & \text{ if } Z_t =0 
\end{cases} 
\end{align*}
for any  $t \notin \mathcal{T}$
\item $M_\tau = \lim_{t \rightarrow \tau} M_t$ for any $\tau \in \mathcal{T}$, and
\[
dM_t \equiv 
\begin{cases}
0 & \text{ if } Z_t <0 \\
\max \left\{ 0, dX_t \right\} & \text{ if } Z_t =0
\end{cases}
\]
for any  $t \notin \mathcal{T}$
\end{enumerate}
It is easy to see that $dM_t \times dZ_t=0$, $M_t$ is increasing, and $Z_t$ verifies at all $\tau \in \mathcal{T}$
\[
 Z_{\tau-}+U_{\tau-}  = U_\tau 
\]
so we have
\begin{equation} \label{jump process}
 Z_t  \equiv X_t-M_t -U_t = Y_t - M_t
\end{equation}

A development policy is now completely defined by the collection of stopping times $\mathcal{T}$ and the processes $Z_t$ and $M_t$. First consider some $t \notin \mathcal{T}$, and assume that $Z_t<0$. In this case, $dZ_t=dX_t$ and $dM_t=0$. If $Z_t=0$ and $dX_t>0$, then  $ dM_t=dX_t>dZ_t=0$, but if $dX_t<0$, then  $ 0 = dM_t>dZ_t=dX_t$. Now consider $\tau \in \mathcal{T}$ which are the points at which $Z_t$ jumps back to $0$ and $dM_\tau=0$.

The process $Z_t$ is a reflected brownian motion at $0$ controlled at each set of stopping times and it is usually referred (its absolute value) as the drawdown process of a (impulse controlled) brownian motion.\footnote{See Proposition \ref{Prop: sup issues} in the Appendix \ref{app} for analysis of this stochastic process and Chapter 2 in \cite{book:harrison2013} for a discussion of this process in the context of a reflected brownian motion.} The process $M_t$ is increasing and traces the supremum of $X_t-Z_t$. Moreover, at every point of continuity, $\lim_{dt \rightarrow 0} \frac{\mathbb{E}(dM_t)}{dt} \rightarrow \infty$ if $Z_t=0$ (see Proposition \eqref{Prop: sup issues}).

More formally
\begin{definition}
A development policy $\Gamma$ is a collection of stopping times $\mathcal{T}$, and a pair of stochastic process, $Z_t$ and $M_t$, such that 
\begin{enumerate}
\item $Z_t$ is non-positive (reflected at 0) and $M_t$ is non-negative and increasing.
\item For all $t \notin \mathcal{T}$, 
\begin{enumerate}
\item if $Z_t<0$, $dZ_t = \mu dt + \sigma dWt$ and $dM_t =0$,
\item if  $Z_t=0$, $dZ_t=\min \{0, \mu dt + \sigma dWt \} $, and $dM_t = dX_t+dZ_t$.
\end{enumerate}
\item For all $\tau \in \mathcal{T}$, $Z_\tau=0$ and $dM_t =0$.
\end{enumerate}
A feasible development policy verifies $\Pr (\sup \mathcal{T}  < \infty | \mathcal{F}_t)=1$ so there is a final intervention in the development process
\end{definition}
With these processes, we define
\begin{definition}
At every $t \geq 0$ the version of product that delivers the highest profits if marketed
\[
 q^*_t = \sup \left\{ Y_s : s \leq t \right\} = M_t 
 \]
\end{definition}

Having defined all interventions while developing the product, we let $T$ be decision to stop the development process for the last time and either launch the product or abort the process altogether. Since a developer will not restart the development process by paying the cost RR while simultaneously stopping development, it must be that $T > \sup \left\{ \mathcal{T} \right\}$. If the developer decides to stop the process without launching the product, she obtains $0$, while if she decides to launch the product she obtains $\pi(M_T)-L$. It follows that the expected utility of a developer with policy $\Gamma= \left\{ \mathcal{T}, \left\{ Z_s \right\}_{s \geq 0} ,  \left\{ M_s \right\}_{s \geq 0} \right\}$ and a launch $T > \sup \left\{ \mathcal{T} \right\}$ is given by
\small
\begin{align} \label{Value General} 
U(\Gamma,T \; | \;  X_0) & = \mathbb{E} \left[  e^{-r T} \left( \max\{  \pi(M_T )- L , 0 \}  \right)   - \int_0^T r e^{-r s} c ds- \sum_{\tau \in  \mathcal{T}  }  e^{-r \tau} R   \; | \; X_0 \right]
\end{align}
\normalsize 

The developer seeks to maximize \eqref{Value General}  by choosing $\Gamma$ and the final decision $T$. In other words, the developer maximizes her expected utility by choosing ``when'' to restart the development of the product from a previous better version and ``when'' to market the product. The state of the system can then be described by $(Z_t,M_t)$ which are functions of the underlying process $X_t$ as described above once the policies are in place. As it is well known $(X_t,M_t)$ is a Markov process (see \cite{peskir1998}), and hence $(Z_{t-\tau},M_{t-\tau})$ is also a Markov process conditional on any stopping time $\tau \in \mathcal{T}$ by the optional stopping time theorem. It follows then, that the expected utility \eqref{Value General} can be written as a function of the relevant states
\small
\begin{align} \label{Value Final} 
 \mathbb{E} \left[U(\Gamma,T) \; | \;  Z_t,M_t \right]  & = \mathbb{E}  \left[  e^{-r (T-t)} \left( \max\{  \pi(M_T )- L , 0 \}  \right) - \int_0^{T-t} r e^{-r s} c ds - \sum_{\tau \in  \mathcal{T}, \tau>t  }  e^{-r (\tau-t)} R \; | \; Z_t,M_t \right] 
\end{align}
\normalsize

\bigskip

\section{Solving the model}

Since $dZ_t \: dM_t = 0$, the problem can be effectively decomposed into two independent dimensions. These dimensions jointly determine the free boundary $z^*(m)$, which represents the threshold at which the developer optimally intervenes. Specifically, for a fixed value of $m$, the developer faces a stopping problem where she must decide when to stop the development process and collect the continuation value that depends on the decision at the stopping time. However, when $m$ changes, the developer’s value for an intervention evolves, leading to a reformulation of the stopping problem. This second step determines the evolution (changes) of $z^*(m)$.

We first prove necessary and sufficient conditions for a function
\[
W: (-\infty,0] \times [0, \infty) \rightarrow \mathbb{R} 
\]
that is twice continuously differentiable in the first dimension and once continuously differentiable in the second dimension to represent the value of the maximal policy, $W(z,m)$, at state $(z,m)$. These conditions are well known: the Bellman-Jacobi-Hamilton equation when $z$ moves, the value-matching condition when the developer intervenes, and the reflection condition when $z = 0$. This is presented in Proposition~\ref{Prop: Ito's lemma}.

The problem is now to find the proper candidate function for $W(\cdot,\cdot)$ and verify that the necessary and sufficient conditions hold. We proceed in steps. First, we use properties of the Wald Martingale to find an explicit solution for the value function $W(z, m)$ within an arbitrary boundary. The problem is now to determine the optimal boundary. Since this is a standard stopping problem, the smooth pasting condition and the value matching condition provide the properties of the value function at the boundary. Using the explicit solution of the value function, we derive a new expression for it in terms of $(z,m)$, the boundary $z^*(m)$, and the continuation value at the boundary.

The magnitude of $m$ determines the nature of the continuation value. For relatively small values of $m$, the boundary $z^*(m)$ dictates whether the developer should restart the development process or abandon it entirely. Conversely, for sufficiently large values of $m$, the boundary determines whether the developer should restart development or proceed with launching the product. Applying the reflection condition, we find that the continuation values are ordered in a simple way along the $m$ dimension: for low values of $m$ the developer aborts the project, for intermediate values of $m$ the developer restarts the development process, and for high values of $m$ the developer launches the product. This allows us to further refine the value function and express it solely in terms of the cutoffs and $z^*(m)$, and to determine the cutoffs as implicit functions of $z^*(m)$. This is formalized in Proposition~\ref{proposition monotonicity} and Corollary~\ref{Cutoffs: smooth pasting}.

To show that the candidate value function exists, now it suffices to prove that there is a boundary $z^*(m)$ that is well defined, together with the cutoffs that determine the changes in the continuation value. Using the reflection condition, this is equivalent to showing that there exists a solution to a differential equation that is piecewise continuous. Since $dM_t > 0$ implies that $z = 0$, $m$ evolves along the reflected boundary described above and determines the value of the problem at that boundary. Consequently, the reflection condition becomes essential in characterizing the optimal policy: it provides how the free boundary $z^*(m)$ adjusts in response to changes in $m$, $\frac{\partial z^*(m)}{\partial m}$, capturing the sensitivity of the intervention threshold to changes in $m$ and, hence, the continuation value.

Both $z^*(m)$ and its derivative, $\frac{\partial z^*(m)}{\partial m}$, are implicitly defined by the continuation values, which are inherently forward-looking. Although the continuation value during the development process is endogenously determined, it is exogenous once the decision-maker opts to launch the product. This structure allows us to solve the differential equation backwards to find the free boundary, and hence the value function, as a function of expected future performance.

Unfortunately, the differential equation does not have a known explicit solution for two of the three segments. Indeed, when the state is in the intermediate stage, the differential equation resembles a generalized Bernoulli equation, and in the final stage, it resembles an Abel differential equation of the second kind. We study qualitative properties of the implicit solution and discuss its shape to provide a characterization of the free boundary, which we illustrate in Figure~\ref{picture: strategy}.

\bigskip

\subsection{Simplifying the problem} 

In this section, we present the decomposition of the problem into two distinct components that characterize the free boundary $z^*(m)$, as outlined above. Specifically, we first analyze the evolution of the value function of the optimal policy, denoted by $V(z,m)$, within the free boundary for $z \in (0, z^*(m))$. We then examine its behavior on the reflection boundary, $\frac{\partial V(0,m)}{\partial m}$, given an arbitrary continuation function $V(0,m)$. The former condition follows from the Bellman-Hamilton-Jacobi (BHJ) equation which emerges from the application of a generalized version of Itô-Watanabe's lemma (\cite{book:KS1991}), while the latter is derived from the reflection condition (\cite{book:harrison2013})

Consider some $\tau  \in \mathcal{T}$ and note that for any $dt$, sufficiently small, we can write
\[
Z_{\tau+dt} = \int_\tau^{\tau+dt}  \sigma dW_s + \int_\tau^{\tau+dt} \mu ds - \int_{\tau}^{\tau+dt} dM_s   
\]
Since $Z_\tau=0$ for all $\tau \in \mathcal{T}$, the process $Z_t$ is a semimartingale with finite variation process that can be decomposed on its continuous part $\int_0^t \mu ds - \int_{0}^{t} dM_s$ and its left limits-right continuous jumps $-Z_{\tau-}$:
\[
Z_t + M_t = \sigma W_t + \mu t  + \sum_{\tau \in \mathcal{T}}  \left(  -Z_{\tau-} \right)
\]
The optimal choice of these jumps and the final decision to stop the development process determine the optimal policy.

The following proposition, applying a generalized version of Ito's lemma and the derivation of optimal conditions, characterizes necessary and sufficient conditions for the value function (among a certain class) of an optimal policy $(\Gamma,T)$.

\begin{proposition} \label{Prop: Ito's lemma}
Let $V: [-\infty,0] \times [0, \infty) \rightarrow \mathbb{R} $ be a twice continuously differentiable function in its first argument and once continuously differentiable in the second argument, then for all feasible policy $(\mathcal{T},T)$
\[
V(z,m) \geq \mathbb{E} \left[U(\mathcal{T},T) \; | \;  z,m \right] , \: \forall z \leq 0 \leq m,
\]
if and only if
\begin{align*}
 V \left( z,m \right) & = \: \max \left\{  0 ,  \pi(m)-L ,  V \left(  0,m \right)  - R ,   \frac{\mu}{r}  V'_z \left( z,m \right)    +  \frac{\sigma^2}{2r} V''_{zz} \left( z,m \right)   - c \; \right\} \nonumber
\end{align*}
and (reflection condition)
\begin{equation} \label{reflection c}
V'_m \left( 0,m \right)  = V'_z \left( 0,m \right)
\end{equation}

\end{proposition}

The value function obeys standard rules. When 
\[
V \left( z,m \right) > \: \max \left\{  0 ,  \pi(m)-L ,  V \left(  0,m \right)  - R  \right\}
\]
the value function follows the BHJ equation as $z$ moves freely. The reflection condition gives us how the value function changes when we move from one dimension to the other one, when $z$ reflects. To understand it note that when $z=0$ the value function moves in the direction of $m$ when $dx>0$ while it moves in the direction of $z$ when $dx<0$. Abusing some notation we can think of the reflection condition as a differentiability condition:
\[
\frac{dV(\cdot)}{dx}_{dx>0} = \frac{dV(\cdot)}{dx}_{dx<0}
\]
Finally, the value matching condition when the developer intervenes
\[
V \left( z,m \right) = \: \max \left\{  0 ,  \pi(m)-L ,  V \left(  0,m \right)  - R  \right\}
\]
gives the continuation value when $z$ stops to move smoothly.\footnote{The smooth pasting condition is the consequence of differentiability of the value function.}

\subsection{Candidate Value Function.}

Using properties of the Wald martingale (\cite{book:stokey2008}), we characterize a candidate for the value function of an optimal policy for a fixed value of $m$, a fixed and arbitrary free boundary $\underline{z}$, and arbitrary decisions at the boundary. With some abuse of notation we still denote it by $W(z,m)$. The reflection condition also implies that $W(0,m)$ is increasing in $m$, ensuring that continuation values are ordered in an intuitive way. This gives directly the changes in the value matching condition, and, as a corollary of differentiability of the value function, the smooth pasting condition. These two conditions together allow us to determine implicitly the free (optimal) boundary $z^*(m)$. The conditions established on the reflected boundary given by equation \eqref{reflection c} yield a characterization of the first-order differential equation governing $z^*(m)$. Using the expression for $W(z,m)$ as a function of the free boundary, we obtain a well-defined first-order, nonlinear, nonhomogeneous differential equation that is piecewise differentiable for $z^*(m)$. To finish we prove that the candidate value function verifies the necessary and sufficient differentiable properties.

 \medskip
 
Consider a state $(z,m)$ with $z<0$ and let $\mathbb{T}({z})$ be the stopping time
\[
\mathbb{T}_{\underline{z}}({z}) \equiv \inf \left\{ t \geq 0 : Z_t \notin (\underline{z},0) \right\}
\]
for some $\underline{z}<z$. The value function verifies
\[
W(z,m) = \mathbb{E} \left[ - \int_0^{\mathbb{T}_{\underline{z}}({z})} r e^{-rs}c ds + e^{-r \mathbb{T}_{\underline{z}}({z})} W(Z_{\mathbb{T}_{\underline{z}}({z})},m) \: | \: z,m \right]
\]
Standard arguments for Brownian motion and (Wald) martingales (see \cite{book:stokey2008}, Chapter 5) give an expression for this function
\begin{align} \label{Solution Wald}
W(z,m) + c& =   (W(0,m)+c)  \frac{e^{\gamma_2 \underline{z} } e^{\gamma_1 z }-e^{\gamma_1 \underline{z} } e^{\gamma_2 z } }{e^{\gamma_2 \underline{z} } -e^{\gamma_1 \underline{z} }}  +  \left( W(\underline{z},m) +c  \right) \frac{e^{\gamma_2 z }-e^{\gamma_1 z }  }{e^{\gamma_2 \underline{z} } -e^{\gamma_1 \underline{z} } } 
\end{align}
for\footnote{For future reference we have that the roots verify
\begin{align*}
\frac{\gamma_1 - \gamma_2}{2} = \sqrt{\left( \frac{\mu}{\sigma^2} \right)^2 + \frac{2r}{\sigma^2}}
\quad \quad \quad \frac{\gamma_1 + \gamma_2}{2} = - \frac{\mu}{\sigma^2} 
%\gamma_1 = - \frac{\mu}{\sigma^2} + \frac{\gamma_1 - \gamma_2}{2}
\quad \quad \quad \gamma_1 \gamma_2 %=  \left(\frac{\mu}{\sigma^2}\right)^2 - \left(\frac{\gamma_1 - \gamma_2}{2} \right)^2 
=- \frac{2r}{\sigma^2}
%\gamma_2 = - \frac{\mu}{\sigma^2} - \frac{\gamma_1 - \gamma_2}{2}
\end{align*}} 
\[
\gamma_1 \equiv - \frac{\mu}{\sigma^2} + \sqrt{\left( \frac{\mu}{\sigma^2} \right)^2 + \frac{2r}{\sigma^2}} >0 >  - \frac{\mu}{\sigma^2} - \sqrt{\left( \frac{\mu}{\sigma^2} \right)^2 + \frac{2r}{\sigma^2}} \equiv \gamma_2 
\]

It is easy to see that \eqref{Solution Wald} verifies the BJH equation
\begin{align} \label{diff1}
 r (W \left( z,m \right)+c)   = \mu W'_z \left( z,m \right)    + \frac{\sigma^2}{2} W''_{zz} \left( z,m \right)  
 \end{align}
 for any $z \in (0,\underline{z})$. Morevoer, we can further define arbitrarily the value function outside the bounds in the following way
 \[
 W(z,m) = W(\underline{z},m), \; \forall z \leq \underline{z}
 \]
This extension allows us to calculate first order conditions at the boundary.

The stopping time $\mathbb{T}_{\underline{z}}({z})$ determines the state at which the developer makes a decision: restarting the development process, with a continuation value $W(0,m)  - R$; launching the product, with a continuation value $\pi(m)-L$; or stopping the process altogether without launching the product with continuation value $0$. More formally, the optimal boundary must verify the value matching condition
\begin{equation} \label{Boundary Value}
W(z^*(m),m)  = \max\{ W(0,m)  - R, \pi(m)-L,0\}
\end{equation}
and $W$ must be differentiable at $z^*(m)$, so the smooth pasting condition must also hold
\[
\frac{\partial W(z,m)}{\partial z}_{\mid z=z^*(m)}  = 0
\]
which follows from the fact that $W(z^*(m),m)$ is, in reality, independent of $z$. Using \eqref{Solution Wald}, the value matching and the smooth pasting conditions give
\begin{align*} 
W(0,m)+c  & =   \left( W(z^*(m),m) +c  \right)  \frac{\gamma_1 e^{-\gamma_2 z^*(m) }-\gamma_2 e^{-\gamma_1 z^*(m) }  }{\gamma_1 -\gamma_2  } 
\end{align*}
and replacing back in \eqref{Solution Wald} we get the value function
\begin{align} \label{Solution Waldb}
W(z,m) + c& =    \left( W(z^*(m),m) +c  \right)  g (z-z^*(m))  
\end{align}
where we define
\[
g(x) \equiv \frac{\gamma_1 e^{\gamma_2 x } -\gamma_2 e^{\gamma_1 x}}{\gamma_1  -\gamma_2 } = e^{\frac{\gamma_1+\gamma_2}{2} x } 
 \frac{\gamma_1   e^{-\frac{\gamma_1-\gamma_2}{2} x } -\gamma_2   e^{\frac{\gamma_1-\gamma_2}{2} x}}{\gamma_1  -\gamma_2 } 
\]
which is convex and achieves a minimum at $0$ because $\gamma_1>0>\gamma_2$.

\medskip

The continuation value, $W(z^*(m),m)$, together with the free boundary, $z^*(m)$, fully determine the candidate value function $W(z,m)$ for an optimal policy. At this stage, the reflection condition \eqref{reflection c} becomes fundamental, as it provides a differential equation for $z^*(m)$ given a continuation value $W(z^*(m),m)$. Since the continuation value $W(z^*(m),m)$ depends on the optimal decision—whether to launch, abort, or restart the project—identifying the values of $m$ at which each decision is optimal will allow us to characterize the free boundary.

Moreover, the reflection condition plays a crucial role in determining $W(z^*(m),m)$ even when $z^*(m)$ is not fully characterized. This is because it establishes a direct link between the changes in $W(0,m)$ given by $m$ and the changes in $W(z,m)$ given by $z$ when reflected, ensuring that $W(z,m)$ is increasing in $m$. This property enables us to define the shape of $W(z^*(m),m)$ as a function of $m$ in a tractable manner. The next proposition exploits this structure to delineate three distinct regions of activity.

\begin{proposition} \label{proposition monotonicity}
The optimal policy is described by 1) the continuous and piecewise differentiable function $z^*(m)$, and 2) the cutoffs $0 \leq m_0 < m_1 < m^* = \overline{q}$, that jointly verify
\begin{equation*} 
0 = m_0 \left( g (-z^*(m_0)) - \frac{R+c}{c} \right)    \quad \quad g (-z^*(m_1)) = 1+ \frac{R}{\pi(m_1)-L+c} \quad \quad z^*(m^*) =0
\end{equation*}
and
\begin{align} \label{reflection in prop}
\frac{\partial z^*(m)}{\partial m}  & =  
\begin{cases}
-1 & \text{ if } m < m_0 \\
g (-z^*(m))-1  & \text{ if } m \in [m_0, m_1) \\
\frac{\pi'(m)}{\pi(m)-L+ c}  \frac{g' (-z^*(m))}{g (-z^*(m))}  & \text{ if } m \in [m_1, m^*) \\
0 & \text{ if } m \geq m^* 
\end{cases}
\end{align}
In all states $(z,m)$ that verify $z \in [0, z^*(m))$ the developer does not intervene and at $(z,m)=(z^*(m),m)$ the developer
\begin{itemize}
\item aborts the project if $m <m_0$;
\item restarts the development process if $m \in [m_0,m_1)$;
\item launches the product if $m > m_1$.
\end{itemize}

\end{proposition}

In Figure \ref{picture: strategy} we illustrate the optimal development policy, which corresponds to the behavior of the free boundary. For low values of $ m $, specifically when $ m < m_0 $, the pair $ (z^*(m), m) $ are  the states in which the developer abandons the process and exits. Although the value of reaching the free boundary—representing a true failed attempt at development—remains constant, $W^*(z^*(m),m) = 0$, each subsequent new development moves the process closer to the viability stage starting at $m_0$. As a result, $W(0,m)$ increases, and the developer becomes progressively more tolerant of unsuccessful attempts seeking to materialize the accumulated progress. That is, $z^*(m)$ is decreasing. In this stage, luck plays a crucial role in ensuring survival.

A regime shift occurs when $ m_0 $ is reached: the developer begins to actively engage in the development process. In this intermediate stage, where $ m \in [m_0, m_1) $, the pair $ (z^*(m), m) $ are the states in which the developer restarts the process from $ (0, m) $ with associated payoff given by $W(z^*(m), m) = W(0, m) - R$. At this point, the developer \textit{knows} that she will eventually launch the product, but additional development is required. Two effects coexist in this stage: further developments increase both the value attained at the boundary, $W^*(z^*(m),m)$, and the value of reaching each development milestone, $W(0,m)$. The result of these two opposing effects is that the developer becomes more impatient and progressively less tolerant of setbacks, leading to faster project restarts, though at a diminishing rate: $z^*(m)$ is increasing and concave. In this stage, the launch is inevitable, yet postponed, and luck improves market prospects.

Once the product is ready for market entry, $m > m_1$, the developer makes a final push to enhance its quality.  In this stage, $ (z^*(m), m) $ are the states in which the developer terminates the process and launches the product with associated payoff given by $W(z^*(m), m) = \pi(m) - L$. Further developments continue to increase both the value attained at the boundary, $W^*(z^*(m),m)$, and the value of reaching each development milestone, $W(0,m)$. The developer grows increasingly impatient and less tolerant of failed paths, leading to an increasing function $z^*(m)$. In this stage, luck enhances market prospects while delaying product launch. However, if the development trajectory is sufficiently unfavorable, the developer halts development and launches the product. Here, the developer exhibits a degree of strategic greed, optimizing for the best possible market outcome.

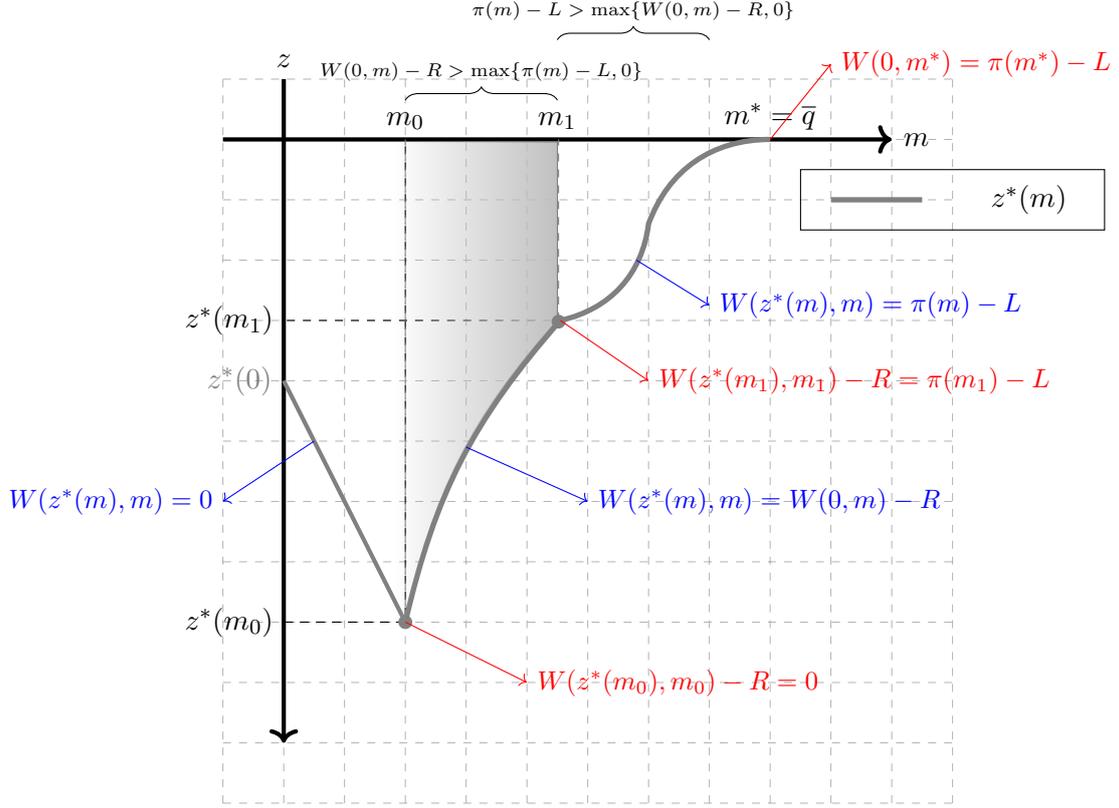
\begin{figure}
\begin{center}
\begin{tikzpicture}[scale=0.8]
% grid
\draw[help lines, color=gray!60, dashed] (-1,1) grid (11,-11);
\draw[->,ultra thick] (-1,0)--(10,0) node[right]{$m$};
\draw[<-,ultra thick] (0,-10)--(0,1) node[above]{$z$};
% add-ons describing state on x-axis
\draw[dashed] (0,-8) -- (2,-8)--(2,0) node[above]{$m_0$} ;
\draw(4.5,0) node[above]{$m_1$};
 \draw(8,0) node[above]{$m^*=\overline{q}$};
% comments on x-axis
\draw [decorate, decoration = {brace,amplitude=5pt}]  (2,0.65) -- (4.5,0.65) node[midway,yshift=1em]{\tiny{$W(0,m)-R >\max\{ \pi(m)-L,0 \}$}};
\draw [decorate, decoration = {brace,amplitude=5pt}]  (4.5,1.65) -- (7,1.65) node[midway,yshift=1em]{\tiny{$\pi(m)-L  >\max\{ W(0,m)-R,0 \}$}};
% add-ons describing cutoffs on y-axis
\draw[ultra thick,color=gray] (2,-8)--(0,-4)  node[left]{$z^*(0)$};
\draw[dashed] (4.4,-3)--(0,-3) node[left]{$z^*(m_1)$};
\draw[dashed] (2,-8)--(0,-8) node[left]{$z^*(m_0)$};
% Boundary function 
\draw[gray, dashed]  (2,-8)--(2,0); \draw[gray, dashed]  (4.52,-3)--(4.52,0); 
\filldraw[gray] (2,-8) circle (3pt); 
\draw [gray,line width=2] (2,-8) to [bend left=14] (4.55,-3) ; 
\draw[gray,fill=lightgray,path fading=west] (1.8,-8)--(1.8,-0.04) -- (4.5,-0.04) -- (4.5,-3) to [bend right=15] (2,-8);
\filldraw[gray] (4.52,-3.02) circle (3pt); 
\draw [gray,line width=2] (4.55,-3) to [bend right=35] (6,-1.4) to [bend left=35] (8,0) ;  
% Boundary function comments
\begin{scope}
\draw [black] (8.5,-1.5) rectangle ++(5,1);
\draw [gray,line width=2] (9,-1) -- (10.5,-1);
\draw (10.5,-1) node[right]{$\quad \quad z^*(m) \quad \quad$};
\end{scope}
\draw [->, blue] (0.5,-5) --(-1,-6)  node[left]{\small{$W(z^*(m),m)=0$}};
\draw [->, blue] (3,-5.1) --(5,-6)  node[right]{\small{$W(z^*(m),m)=W(0,m)-R$}};
\draw [->, blue] (5.8,-2) --(7,-2.75)  node[right]{\small{$W(z^*(m),m)=\pi(m)-L$}};
% Cutoffs comments 
\draw [->, red] (2,-8) --(4,-9)  node[right]{\small{$W(z^*(m_0),m_0)-R= 0$}};
\draw [->, red] (4.55,-3) --(6,-4)  node[right]{\small{$W(z^*(m_1),m_1)-R= \pi(m_1) - L$}};
\draw [->, red] (8,0) --(9,1.25)  node[right]{\small{$W(0,m^*)= \pi(m^*)-L $}};
\end{tikzpicture}
\end{center}
\caption{Development strategy when $m_0>0$: the free boundary $z^*(m)$ in thick gray. In states $(z^*(m),m)$ for $m<m_0$ the project is terminated without launching the product.  In states $(z^*(m),m)$ for $m \in [m_0,m_1)$, the development project is restarted at $(0,m)$. In states $(z^*(m),m)$ for $m \geq m_1$, the product is launched obtaining a payoff of $\pi(m) -L $. 
\label{picture: strategy}}
\end{figure}

The following corollary to Proposition \ref{proposition monotonicity} establishes additional properties of the solution by expressing the value function as implicit function of the optimal free boundary (policy):
\begin{corollary} \label{Cutoffs: smooth pasting}
The value function $W(z,m) $ is increasing in $z$ and increasing in $m$ along the boundary $(0,m)$ and verifies
\begin{align} \label{Solution Waldc} 
W(z,m) + c & =  
\begin{cases}
c   g (z-z^*(m)) & \text{ if } m < m_0 \\
\frac{R}{g (-z^*(m))-1} g (z-z^*(m))   & \text{ if } m \in [m_0, m_1) \\
(\pi(m)-L+c)  g (z-z^*(m))   & \text{ if } m \in [m_1, m^*) \\
\pi(m)-L +c & \text{ if } m \geq m^* 
\end{cases}
\end{align}
for the cutoffs $0 \leq m_0 < m_1 < m^* = \overline{q}$ and the boundary function $z^*(m)$ described in Proposition \ref{proposition monotonicity}
\end{corollary}

\medskip

\subsection{The impact of market size}

We have thus far remained agnostic regarding $m_0 > 0$ or not. In fact, $ m_0 $ is determined by ``walking'' backwards along the free boundary $ z^*(m) $. A direct inspection of Figure \ref{picture: strategy} reveals that, starting from $ m^* $, we have
\[
\int_{m_1}^{\overline{q}}  \frac{\partial z^*(x)}{\partial x} \, dx = -z^*(m_1),
\]
and from $ m_1 $ to $ m_0 $,
\[
\int_{m_0}^{m_1}  \frac{\partial z^*(x)}{\partial x} \, dx = z^*(m_1) - z^*(m_0).
\]
Combining these with Proposition \ref{proposition monotonicity}, we obtain
\[
m_0 = \overline{q} - z^*(m_0) - \int_{m_0}^{m_1} g(z^*(x)) \, dx - \int_{m_1}^{q} \frac{\pi'(x)}{\pi(x) - L + c} \cdot \frac{g(z^*(x))}{g'(z^*(x))} \, dx.
\]
Whether $ m_0 > 0 $ depends on the value of $ \overline{q} $, which aligns with intuition: if $ \overline{q} $ is relatively high compared to $ R $ and $ L $, the expected project's value is substantial, and the developer is incentivized not to abort the project. Thus, the exploration stage disappears and $ m_0 = 0 $. Conversely, if $ \overline{q} $ is low, the developer may prefer to terminate the project if failures piled up, resulting in $ m_0 > 0 $.

This intuitive link between $ \overline{q} $ and $ m_0 $ highlights the need to first understand how variations in $ \overline{q} $ affect the developer's incentives. The following lemma, illustrated in Figure \ref{picture: impact of q}, provides this intermediate step.

\begin{lemma}\label{Lemma comparative results c}

The optimal boundary $z^*(m)$ is decreasing in $\overline{q}$ at all $m < m_0$  and  $m > m_1$, increasing at all $m \in ( m_0 , m_1 )$. The cutoff  $m_0$ increases with $\overline{q}$ while $z^*(m_0)$ does not change, and the changes in $m_1$ and $z^*(m_1)$ are given by
\begin{align} \label{effects of q}
  \frac{ d m_1}{d \overline{q}}  & = \left( \frac{ g'(-z^*(m_1)) }{\frac{\pi'(m_1)}{\pi(m_1)-L+c}  - g'(-z^*(m_1))} \right)  e^{\int_{m_1}^{\overline{q}} h(x) dx }  \lim_{m \uparrow \overline{q}} \left( \frac{\partial z^*(m)}{\partial m}  \right) \\
%    \frac{ d z^*(m_1)}{d \overline{q}}  & =  \left(   \frac{\frac{\pi'(m_1)}{\pi(m_1)-L+c} (g(-z^*(m_1))-1)}{\frac{\pi'(m_1)}{\pi(m_1)-L+c} -g'(-z^*(m_1))}  \right) e^{\int_m^{\overline{q}} h(x) dx }  \lim_{m \uparrow \overline{q}} \left( \frac{\partial z^*(m)}{\partial m}  \right) \nonumber \\
      \frac{ d z^*(m_1)}{d \overline{q}}      & =      \frac{\pi'(m_1)}{\pi(m_1)-L+c} \frac{g(-z^*(m_1))-1}{g'(-z^*(m_1))} \frac{ d m_1}{d \overline{q}} 
        \nonumber
 \end{align}

\end{lemma}

\begin{figure}[h]
\centering
\begin{subfloat}[$\underset{m \uparrow m_1}{ \lim} \frac{\partial z^*(m)}{\partial m}>\underset{m \downarrow m_1}{ \lim}  \frac{\partial z^*(m)}{\partial m}$ or $g'(-z^*(m_1)) >  \frac{\pi'(m_1)}{\pi(m_1)- L+ c}$.\label{picture: comparativea}] 
{
\begin{tikzpicture}[scale=0.57]
% grid
\draw[help lines, color=gray!60, dashed] (-1,1) grid (10,-10);
\draw[->,ultra thick] (-1,0)--(10,0) node[right]{$m$};
\draw[<-,ultra thick] (0,-10)--(0,1) node[above]{$z$};
% add-ons describing state on x-axis
\draw[dashed] (0,-8) -- (2,-8)--(2,0) node[above]{\tiny{$m_0$}} ;
\draw(4.5,0) node[above]{\tiny{$m_1$}};
 \draw(8,0) node[above]{\tiny{$m^*=\overline{q}$}};
% comments on x-axis
\draw[dashed] (4.4,-2)--(0,-2) node[left]{\tiny{$z^*(m_1)$}};
\draw[dashed] (2,-8)--(0,-8) node[left]{\tiny{$z^*(m_0)$}};
% Boundary function 
\draw[ultra thick,color=gray] (2,-8)--(0,-4)  node[left]{\tiny{$z^*(0)$}}; 
\draw[gray, dashed]  (2,-8)--(2,0); \draw[gray, dashed]  (4.52,-2)--(4.52,0); \filldraw[gray] (2,-8) circle (3pt); 
\draw [gray,line width=2] (2,-8) to [bend left=14] (4.55,-2) ; 
\draw[gray,fill=lightgray,path fading=west] (1.8,-8)--(1.8,-0.04) -- (4.5,-0.04) -- (4.5,-2) to [bend right=15] (2,-8);
\filldraw[gray] (4.52,-2.02) circle (4pt); 
\draw [gray,line width=2] (4.55,-2) to [bend right=20] (6,-1.4) to [bend left=25] (8,0) ;  
% Changes due to q 
\draw[dashed, ultra thick,color=gray] (1.5,-8)--(0,-5); \filldraw[gray] (1.5,-8) circle (3pt);
\draw[dashed, ultra thick,color=gray] (1.5,-8) to [bend left=20] (3.5,-2.8) ; \filldraw[gray] (3.5,-2.8) circle (3pt);
\draw[dashed, ultra thick,color=gray] (3.5,-2.8) to [bend right=20] (6,-2) to [bend left=25] (9,0) ; 
\end{tikzpicture}
}
\end{subfloat}
\hspace{10pt} 
\begin{subfloat}[$\underset{m \uparrow m_1}{ \lim} \frac{\partial z^*(m)}{\partial m}<\underset{m \downarrow m_1}{ \lim}  \frac{\partial z^*(m)}{\partial m}$ or $g'(-z^*(m_1)) <  \frac{\pi'(m_1)}{\pi(m_1)- L+ c}$.\label{picture: comparativeb}] 
{
\begin{tikzpicture}[scale=0.55]
% grid
\draw[help lines, color=gray!60, dashed] (-1,1) grid (10,-10);
\draw[->,ultra thick] (-1,0)--(10,0) node[right]{$m$};
\draw[<-,ultra thick] (0,-10)--(0,1) node[above]{$z$};
% add-ons describing state on x-axis
\draw[dashed] (0,-8) -- (2,-8)--(2,0) node[above]{\tiny{$m_0$}} ;
\draw(4.5,0) node[above]{\tiny{$m_1$}};
 \draw(8,0) node[above]{\tiny{$m^*=\overline{q}$}};
% comments on x-axis
\draw[dashed] (4.4,-3)--(0,-3) node[left]{\tiny{$z^*(m_1)$}};
\draw[dashed] (2,-8)--(0,-8) node[left]{\tiny{$z^*(m_0)$}};
% Boundary function 
\draw[ultra thick,color=gray] (2,-8)--(0,-4)  node[left]{\tiny{$z^*(0)$}};
\draw[gray, dashed]  (2,-8)--(2,0); \draw[gray, dashed]  (5.02,-3)--(5.02,0); 
\filldraw[gray] (2,-8) circle (3pt); 
\draw [gray,line width=2] (2,-8) to [bend left=14] (5.05,-3) ; 
\draw[gray,fill=lightgray,path fading=west] (1.8,-8)--(1.8,-0.04) -- (5.05,-0.04) -- (5.05,-3) to [bend right=15] (2,-8);
\filldraw[gray] (5.02,-3.02) circle (4pt); 
\draw [gray,line width=2] (5.05,-3) to [bend left=20] (6,-1.4) to [bend right=5] (8,0) ;  
% Changes due to q 
\draw[dashed, ultra thick,color=gray] (1.5,-8)--(0,-5); \filldraw[gray] (1.5,-8) circle (3pt);
\draw[dashed, ultra thick,color=gray] (1.5,-8) to [bend left=20] (6,-2) ; \filldraw[gray] (6,-2) circle (3pt);
\draw[dashed, ultra thick,color=gray] (6,-2) to [bend left=20] (7.5,-1.2) to [bend right=5] (9,0) ; 
 \end{tikzpicture}
}
\end{subfloat}
%\end{center}
\caption{Changes in free boundary $z^*(m)$ (gray) in response to increments 
\label{picture: impact of q}}
\end{figure}
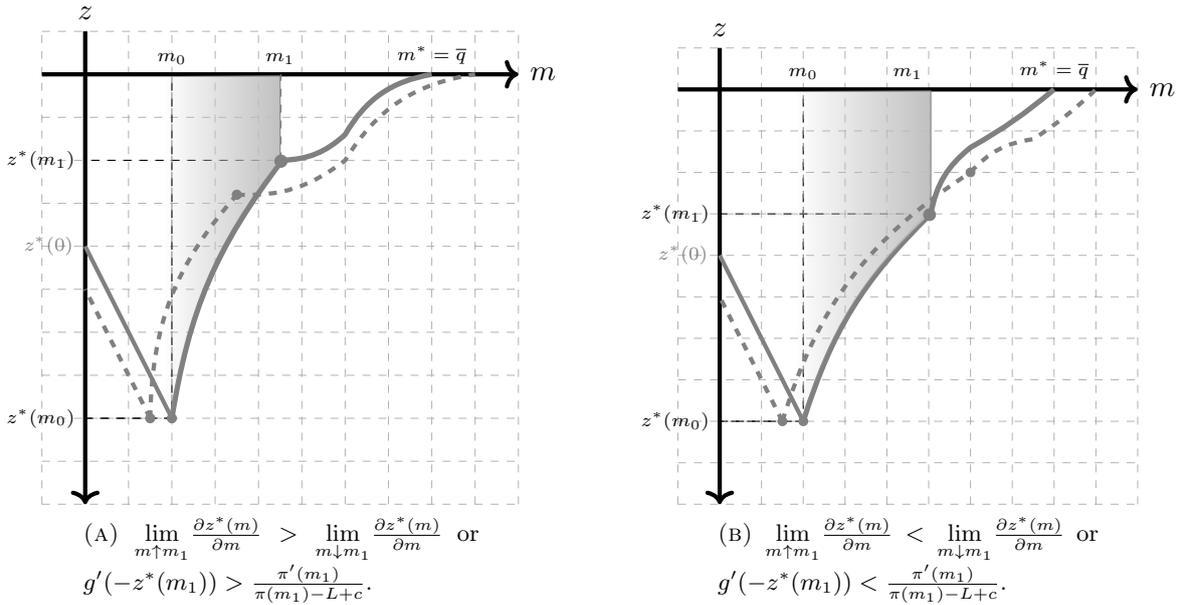

As $ \overline{q} $ increases, the option value of continuing the project in the final stage also increases. This leads to a decrease in $ z^*(m) $ for $ m > m_1 $, or a downward shift of the free boundary in Figure \ref{picture: impact of q}. While the average quality of products launched may increase, the minimum quality is not necessarily affected in a clear direction. Whether $ m_1 $ increases depends on the slope of the free boundary at the state $ (z^*(m_1), m_1) $. The left panel in Figure \ref{picture: impact of q} (Panel A) illustrates a case where $ m_1 $ decreases, whereas the right panel (Panel B) shows an increase in $ m_1 $. This leads to the seemingly paradoxical outcome that, in markets with greater potential profits, the launched product’s quality may actually be lower than in markets with smaller potential profits.

For a developer in the intermediate stage, a higher launch quality implies that delaying the restart of development becomes more costly. Consequently, the option value of waiting to restart the process decreases, which leads to an increase in $ z^*(m) $ for $ m \in (m_0, m_1) $, or an upward shift of the free boundary. In effect, the developer becomes less tolerant of setbacks. Whether this is associated with a larger intermediate stage depends not only on changes in $ m_1 $, but also on how $ m_0 $ responds to changes in $\overline{q}$.

If $ m_0 > 0 $, the initial stage contracts. This is because the increased value of development justifies greater tolerance of early-stage setbacks. As in the final stage, this is reflected in a decrease in $ z^*(m) $ for $ m < m_0 $, appearing as a downward shift in Figure \ref{picture: impact of q}. Thus, if $ m_1 $ increases, the intermediate stage expands—as illustrated in Panel B of Figure \ref{picture: impact of q}. However, if $ m_1 $ decreases, the net effect on the size of the intermediate stage is ambiguous.

We conclude then

\begin{proposition}
If $\overline{q}$ is large, then $m_0=0$ and the exploration stage is absent. If $\overline{q}$ is small, then $m_0>0$ and there is an exploration stage.
\end{proposition}

\section{Active product development} \label{sec: compa}

In this section, we analyze how the parameters $ R $ and $ L $ affect the developer's behavior. We demonstrate that increments in any of these parameters increase the viability threshold $m_0$, thereby raising the likelihood of project termination. The underlying mechanism is nuanced. In the case of $L$, there are two effects to consider: it directly influences the final stage and indirectly propagates through the system via the differential equation. The influence of $ R $ is analogous but less complex, as its direct effect occurs in the intermediate stage.

\medskip

Consider a developer at the beginning of the process:
\[
Z(0)=X(0)=M(0)=0
\]
The product reaches viability at time $T$ that verifies $M(T) > m_0$ but it requires that $Z(t)>z^*(M(t))$ at all $t \leq T$. Therefore, the probability of launch is given by the joint event
\[
\left\{ M(t) > m_0 \right\} \cup \left\{ \min_{s \leq t} \{ Z(s) - z^*(M(s)) \} >0 \right\} 
\]
Since $X_t$ is the underlying stochastic process for both $M_t$ and $Z_t$, and the lower boundary $z^*(m)$ is absorbent in this stage, this problem can be simplified:
\small
\begin{align*}
\Pr \left(M(t) > m_0, \min_{s \leq t} \{ Z(s) - Z^*(M(s)) \} >0 \right) & = \Pr \left(M(t) > m_0, \min_{s \leq t} \{ X(s)-M(s) - Z^*(M(s)) \} >0 \right)  \\
& = \Pr \left(M(t) > m_0, \inf_{s \leq t} \{ X(s) - Z_0 \} >0 \right)  \\
& = \Pr \left(\sup_{s \leq t} \{ X(s)  \}  > m_0, \inf_{s \leq t} \{ X(s)  \} > Z_0 \right) 
\end{align*}
\normalsize
where we used that $Z^*(M(t))=Z_0-M(t)$ since the developer does not intervene. Note that this is nothing but the Gambler's ruin with upper bound $m_0>0$ and lower bound $Z_0<0$ so this can be calculated using the first time hitting distribution (see \cite{book:stokey2008}, chapter 5):

\begin{lemma} \label{Prop: lemma viable}
The probability of viable product
\small
\begin{align*}
\Pr \left(M(t) > m_0, \min_{s \leq t} \{ Z(s) - Z^*(M(s)) \} >0 \right) & =\frac{e^{ -\frac{2 \mu}{\sigma^2}  z^*(0)}- 1}{e^{-\frac{2 \mu}{\sigma^2}  z^*(m_0)}- 1}
\end{align*}
\normalsize
is decreasing in both $R$ and $L$, and increasing in $\overline{q}$.
\end{lemma}

Let's focus on the effects driven by changes in $ R $, as illustrated in Figure~\ref{picture: comparative}. When $ R $ increases, the developer becomes less willing to restart the development process, thereby increasing the option value of waiting. This direct effect is illustrated geometrically by a downward shift in $ z^*(m) $ in the intermediate stage ($m \in (m_0,m_1)$)). The rise in development costs must be offset by either a shorter intermediate stage (i.e., a decrease in $ m_1 $), a higher average launched quality (i.e., an increase in $ m_1 $), or a combination of both. We show that $ m_1 $ can move in either direction, depending on the curvature of $ z^*(m) $ at $ m_1 $. This constitutes the first indirect effect.\footnote{Mathematically, the effect of $R$ is given by
\begin{align*}
\lim_{m \uparrow m_1} \left( 1 + \frac{\partial z^*(m)}{\partial m} \right) &=  \lim_{m \downarrow m_1} \left( 1 + \frac{\partial z^*(m)}{\partial m} \right) + g(-z^*(m_1)) - \frac{g(-z^*(m_1))}{g'(-z^*(m_1))} \frac{\pi'(m_1)}{\pi(m_1)- L+ c} \\
&=  \lim_{m \downarrow m_1} \left( 1 + \frac{\partial z^*(m)}{\partial m} \right) + \frac{g(-z^*(m_1))}{g'(-z^*(m_1))} \left( g'(-z^*(m_1)) -  \frac{\pi'(m_1)}{\pi(m_1)- L+ c} \right)
\end{align*}
and Claim \ref{Claim comparative results} in Appendix \ref{App Main}, the equation that determines whether $m_1$ increases or not, reduces to  
\begin{align*}
\lim_{m \uparrow m_1} \left( 1 + \frac{\partial z^*(m)}{\partial m} \right) > \lim_{m \downarrow m_1} \left( 1 + \frac{\partial z^*(m)}{\partial m} \right) \iff   \frac{d m_1}{d R} >  0 
\end{align*}
}
As the value of development decreases, the incentive to initiate the process also declines. This implies that the value at $(0,0)$ decreases and without improvement the process is aborted with a higher likelihood; i.e.$z^*(0)$ is smaller. This constitutes the second indirect effect and is depicted geometrically by an upward shift in $ z^*(m) $ during the initial stage.  At the same time, the developer postpones entering the intermediate stage is postponed due to a higher cost of development; i.e. $m_0$ increases. Consequently, the exploration stage expands, and the developer becomes more likely to terminate the project at any given state.
 
\medskip

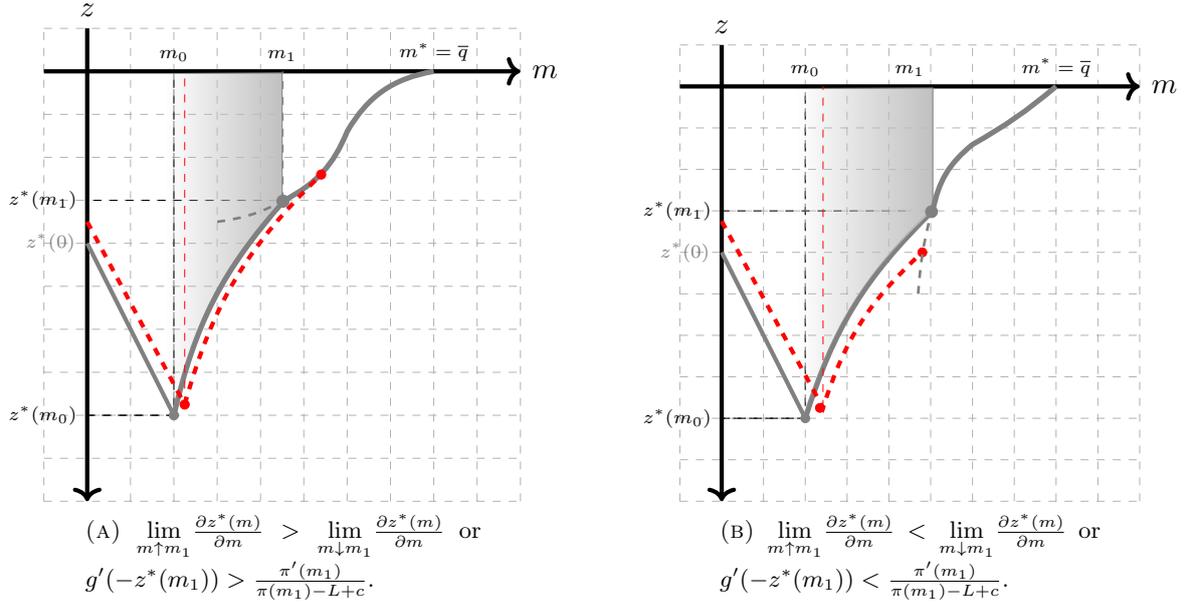
\begin{figure}[h]
\centering
\begin{subfloat}[$\underset{m \uparrow m_1}{ \lim} \frac{\partial z^*(m)}{\partial m}>\underset{m \downarrow m_1}{ \lim}  \frac{\partial z^*(m)}{\partial m}$ or $g'(-z^*(m_1)) >  \frac{\pi'(m_1)}{\pi(m_1)- L+ c}$.\label{picture: comparativea}] 
{
\begin{tikzpicture}[scale=0.57]
% grid
\draw[help lines, color=gray!60, dashed] (-1,1) grid (10,-10);
\draw[->,ultra thick] (-1,0)--(10,0) node[right]{$m$};
\draw[<-,ultra thick] (0,-10)--(0,1) node[above]{$z$};
% add-ons describing state on x-axis
\draw[dashed] (0,-8) -- (2,-8)--(2,0) node[above]{\tiny{$m_0$}} ;
\draw(4.5,0) node[above]{\tiny{$m_1$}};
 \draw(8,0) node[above]{\tiny{$m^*=\overline{q}$}};
% comments on x-axis
\draw[dashed] (4.4,-3)--(0,-3) node[left]{\tiny{$z^*(m_1)$}};
\draw[dashed] (2,-8)--(0,-8) node[left]{\tiny{$z^*(m_0)$}};
% Boundary function 
\draw[ultra thick,color=gray] (2,-8)--(0,-4)  node[left]{\tiny{$z^*(0)$}}; 
\draw[gray, dashed]  (2,-8)--(2,0); \draw[gray, dashed]  (4.52,-3)--(4.52,0); \filldraw[gray] (2,-8) circle (3pt); 
\draw [gray,line width=2] (2,-8) to [bend left=14] (4.55,-3) ; 
\draw[gray,fill=lightgray,path fading=west] (1.8,-8)--(1.8,-0.04) -- (4.5,-0.04) -- (4.5,-3) to [bend right=15] (2,-8);
\filldraw[gray] (4.52,-3.02) circle (4pt); 
\draw [gray,line width=2] (4.55,-3) to [bend right=20] (6,-1.4) to [bend left=25] (8,0) ;  
\draw [dashed, gray,line width=1] (3,-3.5) to [bend right=10] (4.55,-3) ;  
% Change in Boundary function due to R
\draw[dashed, ultra thick,color=red] (0,-3.5) -- (2.25,-7.75) to [bend left=15]  (5.4,-2.4) ; 
\filldraw[red] (5.4,-2.4) circle (3pt); \filldraw[red] (2.25,-7.75) circle (3pt);
\draw[dashed,red,line width=0.25] (2.25,-7.75) -- (2.25,0) ;
\end{tikzpicture}
}
\end{subfloat}
\hspace{10pt} 
\begin{subfloat}[$\underset{m \uparrow m_1}{ \lim} \frac{\partial z^*(m)}{\partial m}<\underset{m \downarrow m_1}{ \lim}  \frac{\partial z^*(m)}{\partial m}$ or $g'(-z^*(m_1)) <  \frac{\pi'(m_1)}{\pi(m_1)- L+ c}$.\label{picture: comparativeb}] 
{
\begin{tikzpicture}[scale=0.55]
% grid
\draw[help lines, color=gray!60, dashed] (-1,1) grid (10,-10);
\draw[->,ultra thick] (-1,0)--(10,0) node[right]{$m$};
\draw[<-,ultra thick] (0,-10)--(0,1) node[above]{$z$};
% add-ons describing state on x-axis
\draw[dashed] (0,-8) -- (2,-8)--(2,0) node[above]{\tiny{$m_0$}} ;
\draw(4.5,0) node[above]{\tiny{$m_1$}};
 \draw(8,0) node[above]{\tiny{$m^*=\overline{q}$}};
% comments on x-axis
\draw[dashed] (4.4,-3)--(0,-3) node[left]{\tiny{$z^*(m_1)$}};
\draw[dashed] (2,-8)--(0,-8) node[left]{\tiny{$z^*(m_0)$}};
% Boundary function 
\draw[ultra thick,color=gray] (2,-8)--(0,-4)  node[left]{\tiny{$z^*(0)$}};
\draw[gray, dashed]  (2,-8)--(2,0); \draw[gray, dashed]  (5.02,-3)--(5.02,0); 
\filldraw[gray] (2,-8) circle (3pt); 
\draw [gray,line width=2] (2,-8) to [bend left=14] (5.05,-3) ; 
\draw[gray,fill=lightgray,path fading=west] (1.8,-8)--(1.8,-0.04) -- (5.05,-0.04) -- (5.05,-3) to [bend right=15] (2,-8);
\filldraw[gray] (5.02,-3.02) circle (4pt); 
\draw [gray,line width=2] (5.05,-3) to [bend left=20] (6,-1.4) to [bend right=5] (8,0) ;  
% Change in Boundary function due to R
\draw[dashed, ultra thick,color=red] (0,-3.25) -- (2.425,-7.75) to [bend left=15] (4.8,-4);
; \filldraw[red] (2.35,-7.75) circle (3pt); \filldraw[red] (4.8,-4) circle (3pt);
\draw[dashed,red,line width=0.25] (2.425,-7.75) -- (2.425,0) ;
\draw[dashed,gray, line width=1] (4.7,-5) to [bend left=5] (5,-3.1);  
\end{tikzpicture}
}
\end{subfloat}
%\end{center}
\caption{Changes in free boundary $z^*(m)$ (gray) in response to increments. \label{picture: comparative}}
\end{figure}

\medskip

We now turn to the effects induced by changes in $ L $, as illustrated in Figure \ref{picture: comparative 2}. An increase in $ L $ makes the developer less inclined to launch a product at a given achieved value. This direct effect is represented geometrically by a downward shift in $ z^*(m) $ during the final stage. As with an increase in $ R $, this reduces the overall value of development, resulting $ z^*(m) $ shifting downwards in the intermediate stage as well. However, two important differences arise. First, the effect on $ m_1 $ and $ z^*(m_1) $ is more intricate due to the interaction of both direct and indirect effects. Second, while $ m_0 $ increases, $ z^*(m_0) $ remains unchanged.

\begin{figure}[h]
\centering
\begin{subfloat}[\label{picture: comparative 2a}] {
\begin{tikzpicture}[scale=0.57]
% grid
\draw[help lines, color=gray!60, dashed] (-1,1) grid (10,-10);
\draw[->,ultra thick] (-1,0)--(10,0) node[right]{$m$};
\draw[<-,ultra thick] (0,-10)--(0,1) node[above]{$z$};
% add-ons describing state on x-axis
\draw[dashed] (0,-8) -- (2,-8)--(2,0) node[above]{\tiny{$m_0$}} ;
\draw(4.5,0) node[above]{\tiny{$m_1$}};
 \draw(8,0) node[above]{\tiny{$m^*=\overline{q}$}};
% comments on x-axis
\draw[dashed] (4.4,-3)--(0,-3) node[left]{\tiny{$z^*(m_1)$}};
\draw[dashed] (2,-8)--(0,-8) node[left]{\tiny{$z^*(m_0)$}};
% Boundary function 
\draw[ultra thick,color=gray] (2,-8)--(0,-4)  node[left]{\tiny{$z^*(0)$}}; 
\draw[gray, dashed]  (2,-8)--(2,0); \draw[gray, dashed]  (4.52,-3)--(4.52,0); \filldraw[gray] (2,-8) circle (3pt); 
\draw [gray,line width=2] (2,-8) to [bend left=14] (4.55,-3) ; 
\draw[gray,fill=lightgray,path fading=west] (1.8,-8)--(1.8,-0.04) -- (4.5,-0.04) -- (4.5,-3) to [bend right=15] (2,-8);
\filldraw[gray] (4.52,-3.02) circle (4pt); 
\draw [gray,line width=2] (4.55,-3) to [bend right=20] (6,-1.4) to [bend left=25] (8,0) ;  
\draw [dashed, gray,line width=1] (3,-3.5) to [bend right=10] (4.55,-3) ;  
% Change in Boundary function due to L
\draw [dashed, ultra thick,color=blue] (0,-3.5) -- (2.5,-8) to [bend left=15] (6.5,-2) to [bend left=25] (8,0)  ;
\filldraw[blue] (2.5,-8) circle (3pt); \filldraw[blue] (6.5,-2) circle (3pt);
\end{tikzpicture}
}
\end{subfloat}
\hspace{10pt} 
\begin{subfloat}[\label{picture: comparative 2b}] 
{
\begin{tikzpicture}[scale=0.55]
% grid
\draw[help lines, color=gray!60, dashed] (-1,1) grid (10,-10);
\draw[->,ultra thick] (-1,0)--(10,0) node[right]{$m$};
\draw[<-,ultra thick] (0,-10)--(0,1) node[above]{$z$};
% add-ons describing state on x-axis
\draw[dashed] (0,-8) -- (2,-8)--(2,0) node[above]{\tiny{$m_0$}} ;
\draw(4.5,0) node[above]{\tiny{$m_1$}};
 \draw(8,0) node[above]{\tiny{$m^*=\overline{q}$}};
% comments on x-axis
\draw[dashed] (4.4,-3)--(0,-3) node[left]{\tiny{$z^*(m_1)$}};
\draw[dashed] (2.5,-9)--(0,-9) node[left]{\tiny{$z^*(m_0)$}};
% Boundary function 
\draw[ultra thick,color=gray] (2,-9)--(0,-5)  node[left]{\tiny{$z^*(0)$}}; 
\draw[gray, dashed]  (2,-9)--(2,0); \draw[gray, dashed]  (4.52,-3)--(4.52,0); \filldraw[gray] (2,-9) circle (3pt); 
\draw [gray,line width=2] (2,-9) to [bend left=14] (4.55,-3) ; 
\draw[gray,fill=lightgray,path fading=west] (1.8,-9)--(1.8,-0.04) -- (4.5,-0.04) -- (4.5,-3) to [bend right=15] (2,-9);
\filldraw[gray] (4.52,-3.02) circle (4pt); 
\draw [gray,line width=2] (4.55,-3) to [bend right=20] (6,-1.4) to [bend left=25] (8,0) ;  
\draw [dashed, gray,line width=1] (3,-3.5) to [bend right=10] (4.55,-3) ;  
% Change in Boundary function due to L
\draw [dashed, ultra thick,color=blue] (0,-4.5) -- (2.5,-9) to [bend left=10] (4.5,-4) to [bend right=20] (6,-2.5) to [bend left=25] (8,0) ;  
\filldraw[blue] (2.5,-9) circle (3pt); \filldraw[blue] (4.5,-4) circle (3pt);
\end{tikzpicture}
}
\end{subfloat}
%\end{center}
\caption{Changes in free boundary $z^*(m)$ (gray) in response to increments in $L$ (blue).
\label{picture: comparative 2}}
\end{figure}
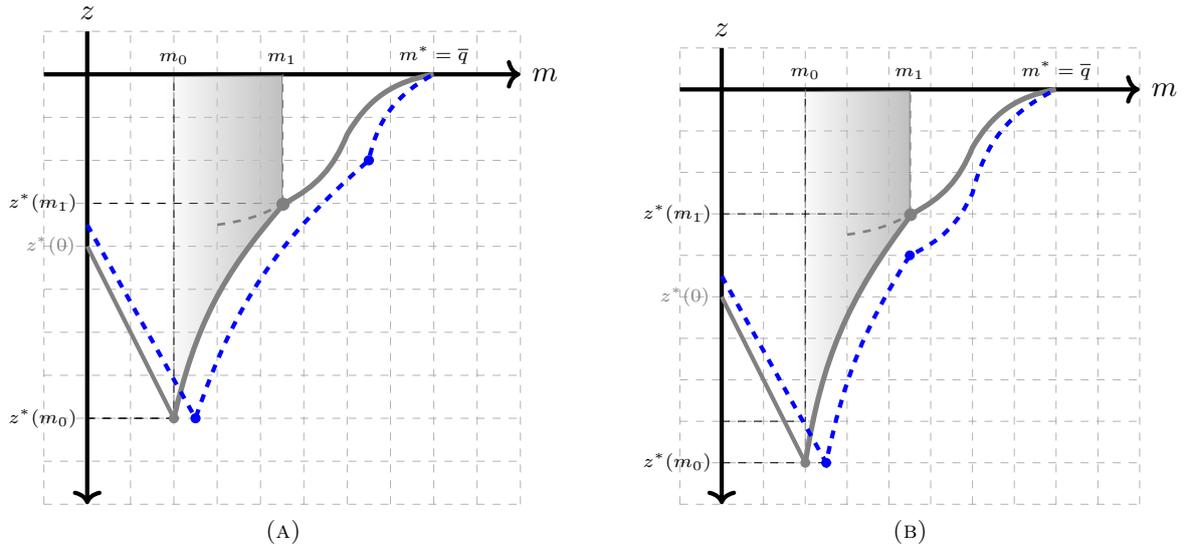

\medskip

\section{Conclusion}

We develop a model of a forward-looking developer who designs a product from idea to market and can actively intervene in its development by restarting the development process from earlier versions, if there is any. The process unfolds in at most three distinct stages. In the first stage, early successes must accumulate to surpass a threshold of viability. If setbacks pile up before this threshold is reached, the developer abandons the project and \textit{quits}. Here, setbacks are indistinguishable from failure, as the product never becomes viable. Once viability is secured, the product enters the second stage of active development. In this stage, setbacks no longer imply failure but are instead reframed as learning opportunities. The developer strategically restarts development from the most promising earlier version, seeking to improve quality. The refusal to stop or settle in the absence of steady progress reflects \textit{grit}—a determined persistence where ``failures'' are seen not as endpoints but as setbacks to be overcome. Nonetheless, luck remains a crucial factor, influencing both the pace of improvement and the total cost of reaching the market. In the third stage, once quality is high enough, the developer refrains from restarting the development process. At this point, the developer becomes \textit{greedy}, pushing for marginal gains in quality in a pursuit of excellence: if improvement stalls and setbacks persist, the developer launches product immediately.\footnote{``Be a yardstick of quality. Some people aren’t used to an environment where excellence is expected.” (Steve Jobs).}

In our model, costlier products to launch may or may not correlate with weaker quality drives; i.e. the length of the active development stage is not necessarily correlated  with the market conditions to launch the product. On the other hand, when the costs to restart development or to launch rise, fewer products are launched and it becomes harder for the developer to reach viability.

Our results align closely with the different forms of funding pursued by developers over the process of product development. Specifically, we highlight the distinction between high-risk, high-reward ventures favored by venture capitalists and lower-risk, market-ready projects typically targeted by private equity investors. In the initial stages of innovation, there is a substantial risk of complete failure, leading potentially to project abandonment—a scenario closely resembling startups funded by venture capital, which often face high failure rates. Once viability is established and market entry becomes achievable, the process transitions into the subsequent development stage. In this stage, restarting the project can be a rational but costly decision, resulting in sunk costs and rendering some pathways unprofitable if too many setbacks occur. Such circumstances typically attract private equity investors, who focus on viable products that successfully enter the market, despite potential ongoing profitability challenges.

Our framework opens multiple avenues for future research, linking the internal dynamics of product development to the broader market context. Our model can shed light on how developers facing market competition replace the current version of a product with a new better one giving a new look at the incumbent-entrant innovation gap (\cite{book:christensen2015}). Before any product launch, the net flow payoff is negative due to development costs and, once a product is on the market, releasing new versions cannibalizes existing ones. The difference between an entrant (no product in the market available) and an incumbent is the stream of income that the latter have and the fact that the marginal benefit of a new version is higher for the former. A key question then is how rapidly these improved versions should be introduced to prevent entry---and whether current losses (where profits are lower than launch costs) might serve as insurance against future setbacks.
 
More generally, our model does not yet account for the effect of different market interactions. Incorporating competition necessitates adjusting the profit function $\pi(m)$ accordingly. In a market with finitely many firms, a rival’s launch generates a discontinuous profit drop, transforming the otherwise deterministic $\pi(m)$ into a jump-down stochastic process. In a fully competitive environment, one could treat $\pi(m)$ as declining over time as new versions of competitive firms enter the market, akin to having different discount factors for costs and for profits. The pace of innovation would then be the outcome of a strategic game, contingent on the market’s particular structure (e.g., number of firms, degree of competitiveness, etc.).

\newpage

 \bibliographystyle{ecta-fullname}
 
\bibliography{DevBiblio}

\newpage

\appendix

\section{Main proofs} \label{App Main}

\subsection*{Proof of Proposition \eqref{Prop: Ito's lemma}.\\}

% \begin{proof}[Proof of Proposition \eqref{Prop: Ito's lemma}]

We first provide a technical result which will be used multiple times.

\begin{lemma} \label{lemma: value of policy}
Let $\Gamma$ an arbitrary policy and $T$ a final decision, and let $F: (-\infty, 0] \times [0, \infty) \rightarrow \mathbb{R}$ be twice continuously differentiable in the first dimension and once continuously differentiable in the second dimension. Then
\small
\begin{align*}
 F \left( z,m \right) & =   \mathbb{E}  \left[ e^{-r t} F \left( Z_{t},M_{t} \right)   \; | \; z,m \right]  - \mathbb{E}  \left[\int_0^{t} e^{-r s}\left( F'_m \left( 0,M_s \right)  - F'_z \left( 0,M_s \right) \right) dM_s  \; | \; z,m \right] \\
&  \quad \quad - \mathbb{E}  \left[ \int_0^{t} e^{-r s} \left( \mu F'_z \left( Z_s,M_s \right)   -  \frac{\sigma^2}{2} F''_{zz} \left( Z_s,M_s \right) -  r  F \left( Z_s,M_s \right)  \right) ds  \; | \; z,m \right]  \\
& \quad \quad  -  \mathbb{E}  \left[ \sum_{\tau < t, \tau \in  \mathcal{T}} e^{- r \tau }  \left( F \left( Z_{\tau},M_{\tau} \right) - F \left( Z_{\tau-},M_{\tau-} \right) \right)  \; | \; z,m \right] \nonumber 
\end{align*}
\normalsize
\end{lemma}
\begin{proof}[Proof of Lemma \ref{lemma: value of policy}]

We apply Ito's lemma to the product $e^{-r t} F \left( Z_t,M_t \right)$ noting that $Z_t$ is a seminartingale, and $M_t$ is a continuous finite variation process (it is increasing) and, second, there are some areas of discontinuity given by the jumps to get  (see \cite{book:harrison2013}, Chapter 4)
\begin{align} \label{Aux a}
e^{-r t} F \left( Z_{t},M_{t} \right) & =  F \left( z,m \right) + \int_0^{t} e^{-r s} dF \left( Z_s,M_s \right) - r \int_0^{t} e^{-r s} F \left( Z_s,M_s \right) ds \\
& \quad \quad \quad \quad + \sum_{\tau < t, \tau \in  \mathcal{T}  } e^{- r \tau }  \left( F \left( Z_{\tau},M_{\tau} \right) - F \left( Z_{\tau-},M_{\tau-} \right) \right)   \nonumber 
\end{align}
Using a generalized version of Ito's lemma (see \cite{book:harrison2013}, Chapter 4.9) for all regions of continuity of $Z_t$ and $M_t$ we have that
\small
\begin{align*}
& dF \left( Z_s,M_s \right)  =  F'_z \left( Z_s,M_s \right) dZ_s + \frac{\sigma^2}{2} F''_{zz} \left( Z_s,M_s \right)  ds + F'_m \left( Z_s,M_s \right) dM_s   \\
& \quad =  F'_z \left( Z_s,M_s \right) dX_s+ \frac{\sigma^2}{2} F''_{zz} \left( Z_s,M_s \right)  ds + \left( F'_m \left( Z_s,M_s \right) - F'_z \left( Z_s,M_s \right) \right)dM_s   \\
& \quad = \left( \mu F'_z \left( Z_s,M_s \right)    + \frac{\sigma^2}{2} F''_{zz} \left( Z_s,M_s \right)  \right) ds + \sigma F'_z \left( Z_s,M \right) dF_s  \\
& \quad \quad \quad \quad + \left( F'_m \left( Z_s,M_s \right) - F'_z \left( Z_s,M_s \right) \right) dM_s
\end{align*}
\normalsize
Replacing in \eqref{Aux a} and taking expectations we have
\small
\begin{align*}
& \mathbb{E}  \left[ e^{-r t} F \left( Z_{t},M_{t} \right)   \; | \; z,m \right]  =  F \left( z,m \right)  + \mathbb{E}  \left[\int_0^{t} e^{-r s}\left( F'_m \left( Z_s,M_s \right)  - F'_z \left( Z_s,M_s \right) \right) dM_s  \; | \; z,m \right] \\
&  \quad \quad + \mathbb{E}  \left[ \int_0^{t} e^{-r s} \left( \mu F'_z \left( Z_s,M_s \right)    + \frac{\sigma^2}{2} F''_{zz} \left( Z_s,M_s \right) -  r  F \left( Z_s,M_s \right)  \right) ds  \; | \; z,m \right]  \\
& \quad \quad  +  \mathbb{E}  \left[ \sum_{\tau < t, \tau \in  \mathcal{T}} e^{- r \tau }  \left( F \left( Z_{\tau},M_{\tau} \right) - F \left( Z_{\tau-},M_{\tau-} \right) \right)  \; | \; z,m \right] \nonumber 
\end{align*}
\normalsize
Recalling that $dM_s \neq 0 $ implies that $Z_s=0$ we get the expression in the lemma.
\end{proof}

We now derive sufficient conditions for the value function of an optimal policy

\begin{lemma} \label{Lemma: nec for optimal}
Let $F: (-\infty, 0] \times [0, \infty) \rightarrow \mathbb{R}$ be twice continuously differentiable in the first dimension and once continuously differentiable in the second dimension. If $F$ verifies the following conditions
\small
\begin{align*}
0 & \leq    F \left( Z_{t},M_{t} \right) - \max \left\{0,\pi (M_T)-L \right\}  \: \forall t \leq T, t \neq \tau    \\
0 & \geq   \left( F'_m \left( 0,M_t \right)  - F'_z \left( 0,M_t \right) \right),\: \forall t \leq T, t \neq \tau \\
0 & \geq     \mu F'_z \left( Z_t,M_t \right)    - \frac{\sigma^2}{2} V''_{zz} \left( Z_t,M_t \right) -  r  \left( F \left( Z_t,M_t \right)   + c \right)   \: \forall t \leq T, t \neq \tau \\
0 & \geq   F \left( Z_{\tau},M_{\tau} \right) - F \left( Z_{\tau-},M_{\tau-} \right) - R , \tau \in  \mathcal{T} \nonumber 
\end{align*}
\normalsize
then
\[
F(z,m) \geq \sup \left\{  \mathbb{E} \left[U(\mathcal{T},T) \; | \;  z,m \right]  \right\},\; \forall \text{ feasible } (\Gamma,T)
\]
\end{lemma}
\begin{proof}[Proof of Lemma \ref{Lemma: nec for optimal}]

Let $(z,m)$ be a state and recall that the expected utility of a policy $\Gamma$ and the final decision made at $T$ is given by
\[
\mathbb{E} \left[U(\mathcal{T},T) \; | \;  z,m \right] = \mathbb{E}  \left[  e^{-r T} \max \left\{0,\pi (M_T)-L \right\} - \left(1- e^{-r T} \right) c  - \sum_{\tau \in \mathcal{T}, \tau < T}  e^{- r \tau } R \; | \; z,m \right]
\]
Applying Lemma \ref{lemma: value of policy} we have that
\small
\begin{align*}
F \left( z,m \right) & = \mathbb{E} \left[U(\mathcal{T},T) \; | \;  z,m \right] + \mathbb{E}  \left[ e^{-r T} F \left( Z_{T},M_{T} \right)   \; | \; z,m \right]    
 - \mathbb{E}  \left[\int_0^{T} e^{-r s}\left( F'_m \left( 0,M_s \right)  - F'_z \left( 0,M_s \right) \right) dM_s  \; | \; z,m \right] \\
&   - \mathbb{E}  \left[ \int_0^{T} e^{-r s} \left( \mu F'_z \left( Z_s,M_s \right)   
 - \frac{\sigma^2}{2} V''_{zz} \left( Z_s,M_s \right) -  r  F \left( Z_s,M_s \right)  \right) ds  \; | \; z,m \right]  \\
&  -  \mathbb{E}  \left[ \sum_{\tau < T, \tau \in  \mathcal{T}} e^{- r \tau }  \left( F \left( Z_{\tau},M_{\tau} \right) - F \left( Z_{\tau-},M_{\tau-} \right) \right)  \; | \; z,m \right] \nonumber \\
& - \mathbb{E}  \left[  e^{-r T} \max \left\{0,\pi (M_T)-L \right\} - \int_0^T r c ds  - \sum_{\tau \in \mathcal{T}, \tau < T}  e^{- r \tau } R \; | \; z,m \right] 
\end{align*}
\normalsize
for $F$ twice continuously differentiable in the first term and once in the second. Rearranging we get
\small
\begin{align*}
F \left( z,m \right) & = \mathbb{E} \left[U(\mathcal{T},T) \; | \;  z,m \right] + \mathbb{E}  \left[ e^{-r T} \left( F \left( Z_{T},M_{T} \right) - \max \left\{0,\pi (M_T)-L \right\}  \right)   \; | \; z,m \right]    \\
&  - \mathbb{E}  \left[\int_0^{T} e^{-r s}\left( F'_m \left( 0,M_s \right)  - F'_z \left( 0,M_s \right) \right) dM_s  \; | \; z,m \right] \\
&   - \mathbb{E}  \left[ \int_0^{T} e^{-r s} \left( \mu F'_z \left( Z_s,M_s \right)   
 - \frac{\sigma^2}{2} V''_{zz} \left( Z_s,M_s \right) -  r  \left( F \left( Z_s,M_s \right)  \right) + c \right) ds  \; | \; z,m \right]  \\
&  -  \mathbb{E}  \left[ \sum_{\tau < T, \tau \in  \mathcal{T}} e^{- r \tau }  \left( F \left( Z_{\tau},M_{\tau} \right) - F \left( Z_{\tau-},M_{\tau-} \right) - R \right)  \; | \; z,m \right] \nonumber 
\end{align*}
\normalsize
It is sufficient for $F \left( z,m \right)  \geq \mathbb{E} \left[U(\mathcal{T},T) \; | \;  z,m \right] $ that
\small
\begin{align*}
0 & \leq   \mathbb{E}  \left[ e^{-r T} \left( F \left( Z_{T},M_{T} \right) - \max \left\{0,\pi (M_T)-L \right\}  \right)   \; | \; z,m \right]    \\
&  - \mathbb{E}  \left[\int_0^{T} e^{-r s}\left( F'_m \left( 0,M_s \right)  - F'_z \left( 0,M_s \right) \right) dM_s  \; | \; z,m \right] \\
&   - \mathbb{E}  \left[ \int_0^{T} e^{-r s} \left( \mu F'_z \left( Z_s,M_s \right)   
 - \frac{\sigma^2}{2} V''_{zz} \left( Z_s,M_s \right) -  r  \left( F \left( Z_s,M_s \right)  \right) + c \right) ds  \; | \; z,m \right]  \\
&  -  \mathbb{E}  \left[ \sum_{\tau < T, \tau \in  \mathcal{T}} e^{- r \tau }  \left( F \left( Z_{\tau},M_{\tau} \right) - F \left( Z_{\tau-},M_{\tau-} \right) - R \right)  \; | \; z,m \right] \nonumber 
\end{align*}
\normalsize
It is sufficient then that the following conditions hold one by one
\small
\begin{align*}
0 & \leq   \mathbb{E}  \left[ e^{-r T} \left( F \left( Z_{T},M_{T} \right) - \max \left\{0,\pi (M_T)-L \right\}  \right)   \; | \; z,m \right]    \\
0 & \leq  - \mathbb{E}  \left[\int_0^{T} e^{-r s}\left( F'_m \left( 0,M_s \right)  - F'_z \left( 0,M_s \right) \right) dM_s  \; | \; z,m \right] \\
0 & \leq   - \mathbb{E}  \left[ \int_0^{T} e^{-r s} \left( \mu F'_z \left( Z_s,M_s \right)   
 - \frac{\sigma^2}{2} V''_{zz} \left( Z_s,M_s \right) -  r  \left( F \left( Z_s,M_s \right)  \right) + c \right) ds  \; | \; z,m \right]  \\
0 & \leq  -  \mathbb{E}  \left[ \sum_{\tau < T, \tau \in  \mathcal{T}} e^{- r \tau }  \left( F \left( Z_{\tau},M_{\tau} \right) - F \left( Z_{\tau-},M_{\tau-} \right) - R \right)  \; | \; z,m \right] \nonumber 
\end{align*}
\normalsize
and since this must hold for all possible $T$ we get the conditions in the lemma 

\end{proof}

We now derive the necessary conditions for an optimal policy. Let
\begin{align} \label{Value of policy}
F \left( z,m \right)  & =  \mathbb{E}  \left[  e^{-r T} \max \left\{0,\pi (M_T)-L \right\} - \left(1- e^{-r T} \right) c  - \sum_{\tau \in \mathcal{T}, \tau < T}  e^{- r \tau } R \; | \; z,m \right] 
\end{align}
Then for $t>0$ with $t < T$ we have the recursive representation
\begin{align} \label{ Aux 0}
F \left( z,m \right)  & =  \mathbb{E}  \left[  e^{-r t} F \left( Z_t,M_t \right) - (1- e^{-r t})   c  - \sum_{\tau \in \mathcal{T}, \tau < t} e^{- r \tau }    R \; | \; z,m \right]
\end{align}
Using again Lemma \ref{lemma: value of policy} and replacing now in \eqref{ Aux 0} we get
\small
\begin{align} \label{Aux 0b}
0 = &   \mathbb{E}  \left[ \int_0^{t} e^{-r s} \left(   \mu F'_z \left( Z_s,M_s,\right)    + \frac{\sigma^2}{2} F''_{zz} \left( Z_s,M_s \right) -  r  \left( F \left( Z_s,M_s \right)  + c \right) \right) ds  \; | \; z,m \right]  \\
 &\quad \quad   +  \mathbb{E}  \left[\int_0^{t} e^{-r s}\left( F'_m \left( Z_s,M_s \right)  - F'_z \left( Z_s,M_s \right) \right) dM_s  \; | \; z,m \right] \nonumber  \\
& \quad \quad  +  \mathbb{E}  \left[ \sum_{\tau < t, \tau \in \mathcal{T}} e^{- r \tau }  \left( F \left( Z_{\tau},M_{\tau} \right) - F \left( Z_{\tau-},M_{\tau-} \right) -R \right)   \; | \; z,m \right] \nonumber 
\end{align}
\normalsize
Let $t' \notin \mathcal{T}$ and the state $(z,m)$ be such that $\Pr (t'+dt \in \left\{ \mathcal{T} \right\} \: | \: z,m)=0$ for small $dt$ so we have that \eqref{Aux 0b} implies
\begin{align} \label{Aux a1}
0 = &   \mathbb{E}  \left[ - \int_0^{t} r e^{-r s} c  + \mu F'_z \left( Z_s,M_s \right)    + \frac{\sigma^2}{2} F''_{zz} \left( Z_s,M_s \right) -  r  F \left( Z_s,M_s \right)  ds \right.  \\
 &\quad \quad   +  \left. \int_0^{t} e^{-r s}\left( F'_m \left( Z_s,M_s \right)  - F'_z \left( Z_s,M_s \right) \right) dM_s  \; | \; z,m \right] \nonumber  
\end{align}
Since both integrals must hold for all $ t' \notin \mathcal{T} $, and $\tau \in \mathcal{T} $ is a random variable, it must be that 
\begin{align} \label{Aux b}
0 & =   \mu F'_z \left(z_{t'},m_{t'} \right)    + \frac{\sigma^2}{2} F''_{zz} \left( z_{t'},m_{t'} \right)  - r \left( F \left( z_{t'},m_{t'} \right)  + c \right) \\
&  \quad \quad \quad +  \mathbb{E}  \left[   \left( F'_m \left( z_{t'},m_{t'} \right) - F'_z \left( z_{t'},m_{t'} \right) \right) \frac{dm_{t'}}{ds}   \; | \; z,m \right] \nonumber
\end{align}
Recall that if $dM_s >0$ then $Z_s=0$ so we have for all $z_{t'}<0$, equation \eqref{Aux b} becomes
\[
0 =  \mu F'_z \left( z_{t'},m_{t'} \right)    + \frac{\sigma^2}{2} F''_{zz} \left( z_{t'},m_{t'} \right)  - r \left( F \left( z_{t'},m_{t'} \right)  + c \right)
 \]
 which is the BJH equation. Since $W$ is twice continuously differentiable in the region of continuity we can take limits as $z_{t'} \uparrow 0$ to obtain the same condition even for  $z_s=0$. It follows then that \eqref{Aux b} becomes for any state $\left( 0,m_s \right)$ for $s \notin \mathcal{T} $, where 
\begin{align*} 
0 & =   \left( F'_m \left( 0,m_{t'} \right)  - F'_z \left( 0,m_{t'} \right) \right) \mathbb{E}  \left[  \frac{dm_{t'}}{ds}   \; | \; z_{t'},m_{t'} \right]
\end{align*}
Proposition \ref{Prop: sup issues} in Appendix \ref{app} we show that $\lim_{ds \rightarrow 0} \mathbb{E}  \left[   \frac{d \sup\left\{Y_z, z\leq s \right\}  }{ds}   \; | \; z,m \right] \rightarrow \infty$ so we must have that  
\[
\lim_{z \uparrow 0} F'_z \left(z,m_{t'} \right) = F'_m \left( 0,m_{t'} \right) 
\]

We now derive the value matching conditions for the intervention as necessary conditions for optimization.  Consider some stopping time $\tau \in \mathcal{T}$ so $Z_{\tau-}<0=Z_{\tau} $ and $M_{\tau} = M_{\tau-}$. Let $V$ the value function of the maximal policy. It follows that for some small $t'>\tau$ it must be that waiting is not optimal

\scriptsize
\begin{align*}
V\left( 0,M_\tau \right) - R & \geq - \int_\tau^{t'} r e^{-r(s-\tau)}c ds +\mathbb{E} \left[  e^{-r(t'-\tau)} V\left( Z_{t'},M_{t'} \right) \; | \; Z_\tau,M_\tau \right] \\
& \geq - \int_\tau^{t'} r e^{-r(s-\tau)} c ds + \mathbb{E} \left[  e^{-r(t'-\tau)} V\left( Z_\tau+dZ,M_\tau \right) \; | \; Z_\tau,M_\tau \right]  \\
& \geq - \int_\tau^{t'} r e^{-r(s-\tau)} c ds +  V\left( Z_\tau,M_\tau \right) \\
& \quad \quad \quad + \mathbb{E} \left[  \int_\tau^{t'} e^{-r(s-\tau)} d V\left( Z_s,M_\tau \right)- r \int_\tau^{t'} e^{-r(s-\tau)} V \left( Z_s,M_\tau \right) ds \; | \; Z_\tau,M_\tau \right]  \\
& \geq - \int_\tau^{t'} r e^{-r(s-\tau)}c ds +  V \left( Z_\tau,M_\tau \right) \\
& \quad \quad \quad + \mathbb{E} \left[   \int_\tau^{t'} e^{-r(s-\tau)} \left( V'_z\left( Z_s,M_{\tau} \right)  dZ_s + \frac{\sigma^2}{2} V''_{zz}\left( Z_s,M_\tau \right)  ds   \right) \right. \\
& \quad \quad \quad  \left. - r \int_\tau^{t'} e^{-r(s-\tau)} V\left( Z_s,M_\tau \right) ds  \; | \; Z_\tau,M_\tau \right] \\
& \geq V\left( Z_\tau,M_\tau \right)  + \mathbb{E} \left[  - \int_\tau^{t'} e^{-r(s-\tau)}  c + \mu V'_z\left( Z_s,M_\tau \right)    + \frac{\sigma^2}{2} V''_{zz}\left( Z_s,M_\tau \right)   ds \right. \\
& \quad \quad   \left.   -  \int_\tau^{t'} r e^{-r(s-\tau)}  V\left( Z_s,M_\tau  \right)    ds +  \int_\tau^{t'} e^{-r(s-\tau)} \sigma  V'_z\left( Z_s,M_\tau \right)  dW_s \; | \; Z_\tau,M_\tau \right] 
\end{align*}
\normalsize
where we used that for $Z_\tau<0$ we must have $d M_\tau=0$. Using the BJH equation and the fact
\[
0 = \mathbb{E} \left[  \int_\tau^{t'} e^{-r(s-\tau)} \sigma  V'_z\left( Z_s,M_\tau \right)  dW_s \; | \; Z_\tau,M \right] 
\]
we get 
\[
V\left( 0,M_\tau \right) - R \geq V\left( Z_\tau,M_\tau \right) 
\]
At the same time, an intervention that occurs before the stopping time $\tau$ is not optimal so for $t'<\tau$ with $Z_{t'}<0$ we must have that

\scriptsize
\begin{align*}
V\left( 0,M_{t'} \right) - R & \leq - \int^{t'+dt}_{t'} r e^{-r(s-t')} c ds + \mathbb{E} \left[  e^{-r dt} V(\left( Z_{t'+dt},M_{t'+dt} \right) \; | \; Z_{t'},M_{t'} \right] \\
& \leq - \int^{t'+dt}_{t'} r e^{-r(s-t')} c ds + V\left( Z_{t'},M_{t'} \right)  \\
&  \quad \quad + \mathbb{E} \left[  \int^{t'+dt}_{t'} e^{-r (s-t')} d V\left( Z_{s},M_{t'} \right) - r\int^{t'+dt}_{t'} e^{-r (s-t')} V\left( Z_{s},M_{t'} \right) ds \; | \; Z_{t'},M_{t'} \right] \\
 & \leq - \int^{t'+dt}_{t'} r e^{-r(s-t')} c ds + V\left( Z_{t'},M_{t'} \right)    + \mathbb{E} \left[  \int^{t'+dt}_{t'} e^{-r (s-t')} \left( V'_z\left( Z_{s},M_{t'} \right) dZ_{s} + \frac{\sigma^2}{2} V''_{zz}\left( Z_{s},M_{t'} \right)  ds  \right)   \right. \\
&  \quad \quad  - \left.  \int^{t'+dt}_{t'} r e^{-r (s-t')} V\left(Z_{s},M_{t'}\right) ds   \; | \; Z_{t'},M_{t'} \right] \\
 & \leq - \int^{t'+dt}_{t'} r e^{-r(s-t')}c ds + V\left( Z_{t'},M_{t'} \right)  \\
& \quad \quad  + \mathbb{E} \left[  \int^{t'+dt}_{t'} e^{-r (s-t')} \left( \mu V'_z\left( Z_{s},M_{t'} \right)  + \frac{\sigma^2}{2} V''_{zz}\left( Z_{s},M_{t'} \right)  -r  V\left( Z_{s},M_{t'} \right)  \right)  ds  \right. \\
&  \quad \quad  \left. +\int^{t'+dt}_{t'}  e^{-r (s-t')} \sigma V'_z\left( Z_{s},M_{t'} \right) dV_{s}  \; | \; Z_{t'},M_{t'} \right] 
\end{align*}
\normalsize
Using again the BJH equation and the fact 
\[
0 = \mathbb{E} \left[  \int_\tau^{t'} e^{-r(s-\tau)} \sigma  V'_z\left( Z_{s},M_{t'} \right)  dW_s \; | \; Z_{t'},M_{t'}\right] 
\]
we get 
\begin{align*}
V\left( 0,M_{t'} \right) - R & \leq  V\left( Z_{t'},M_{t'} \right)  
\end{align*}
It follows then that the collection of stopping times are given by the states
\[
\tau \equiv \inf \left\{ t > 0 : W\left( 0,M_{t} \right) - R  \geq  W\left( Z_{t},M_{t} \right)  \right\}
\]
or for short 
\[
V\left( 0,M_\tau \right) - R = V \left( Z_{\tau-},M_\tau \right) 
\]

% \end{proof}

\qed

\medskip

%\newpage

\subsection*{Proof of Proposition \ref{proposition monotonicity}.\\}

%\begin{proof}[Proof of Proposition \ref{proposition monotonicity}]

\begin{lemma} \label{Form of value F}
The candidate value function is given by \eqref{Solution Waldc}: 
\begin{equation*}  
W(z,m) + c =  
\begin{cases}
c   g (z-z^*(m)) & \text{ if } m < m_0 \\
R \frac{g (z-z^*(m))}{g (-z^*(m))-1}  & \text{ if } m \in [m_0, m_1) \\
(\pi(m)-L+ c)  g (z-z^*(m))  & \text{ if } m \in [m_1, m^*) \\
 \pi(m)-L +c & \text{ if } m \geq m^* 
\end{cases}
\end{equation*}
\end{lemma}

\begin{proof}[Proof of Lemma \ref{Form of value F}]
Using \eqref{Solution Waldb} we get
\[
\frac{\partial W(z,m)}{\partial z} =    \left( W(z^*(m),m) +c  \right)  \gamma_1 \gamma_2 \frac{e^{\gamma_2 (z-z^*(m)) } - e^{\gamma_1 (z-z^*(m)) }}{\gamma_1  -\gamma_2 }   >0
\]
which follows by noting that $\gamma_1>0> \gamma_2 $ and $z>z^*(m)$. Using the reflection condition we have that
\[
\frac{\partial W(0,m)}{\partial m} = \frac{\partial W(z,m)}{\partial z}_{\mid z = 0} >0
\]
Let $\tilde{m}_1$ be defined as $\pi(\tilde{m}_1) = L - R$ and recalling expression \eqref{Boundary Value} and the monotonicity of $W(0,m)$ we get that
\[
W(z^*(m),m)   = 
\begin{cases}
 \max \left\{ 0, W(0,m)  - R \right\}  & \text{ if } m \leq  \tilde{m}_1  \\
 \max \left\{ \pi(m)-L, W(0,m)  - R \right\}  & \text{ if } m >  \tilde{m}_1 
\end{cases}
\]

\noindent \textbf{Case 1:} Assume that there is no $m>0$ such that $W(0,m)  < R$. Recalling the expression \eqref{Boundary Value} it is immediate to see that
\[
W(z^*(m),m)  = 
\begin{cases}
 W(0,m)  - R & \text{ if } m \in [ 0, m_1) \\
 \pi(m)-L & \text{ if } m \geq m_1
\end{cases}
\]
where 
\[
\pi(m_1)-L = W(0,m_1)-R 
\]
and note that it must be that $m_1> \tilde{m}_1$.  It follows then that \eqref{Solution Waldb} and some algebra give that
\begin{align*} 
W(z,m) + c & =  
\begin{cases}
R \frac{ g (z-z^*(m))}{g (-z^*(m)) -1}      & \text{ if } m \in [0, m_1) \\
(\pi(m)-L+c)  g (z-z^*(m))  & \text{ if } m \in [m_1, m^*) \\
 \pi(m)-L+c  & \text{ if } m \geq m^* 
\end{cases}
\end{align*}
Note that $m^*$ then must verify that $g (-z^*(m^*))=1$ and hence $z^*(m^*)=0$. On the other hand, as long as $W(0,m) >  \pi(m)-L$ the developer continues developing the project. First, we are going to show that $m^*<\infty$. Assume this is not the case. Since $\pi(m) \leq \pi( \overline{q})$ it must be that for all $m >  \overline{q}$ we have that $W(z,m) \geq \pi( \overline{q})-L$  and therefore
\[
W(0,m) \geq  \pi( \overline{q}) - L
\]
Consider now another state $m'<m$ but $m'>\overline{q}$ and it must be that $W(0,m)<W(0,m')$
\[
W(0,m) \leq  \pi( \overline{q}) - L
\]
so
\[
W(0,m') >  \pi( \overline{q}) - L
\]
But this is not possible as $\pi( \overline{q}) - L$ is the maximal possible profits. 

Assume to reach a contradiction that $m^*<\overline{q}$, and recall that
\begin{align*} 
W(z,m) + c & =  
\begin{cases}
(\pi(m)-L+ c)  g (z-z^*(m))  & \text{ if } m \in [m_1, m^*) \\
 \pi(m)-L +c & \text{ if } m \geq m^* 
\end{cases}
\end{align*}
Note that we must have that
\[
W(0,m^*)  =  \pi(m^*)-L < \pi(\overline{q})-L
\]
Consider an alternative plan in which the product is not launched at $(0,m^*)$ but it continue to be developed up until some $\hat{z}>0$. Note that in this case, we have the value function
\[
\tilde{W}(z,m^*) + c = (\pi(m^*)-L+ c)  g (z-\hat{z}) > (\pi(m^*)-L+ c) \quad \rightarrow \quad \tilde{W}(z,m^*) > \pi(m^*)-L = W(0,m^*)
\]
contradicting that $W(0,m^*)$ is optimal.

\noindent \textbf{Case 2:} Assume that there is some $m_0>0$ such that for all $m <m_0$ we have that  $W(0,m)  < R$. Recalling the expression \eqref{Boundary Value} it is immediate to see that
\[
W(z^*(m),m)  = 
\begin{cases}
 0 & \text{ if } m< m_0  \\
 W(0,m)  - R & \text{ if } m \in [ m_0, m_1) \\
 \pi(m)-L & \text{ if } m \geq m_1
\end{cases}
\]
where 
\[
\pi(m_1)-L = W(0,m_1)-R 
\]
and note that it must be that $m_1> \tilde{m}_1$. It follows then that \eqref{Solution Waldb} and some algebra give the expressions in \eqref{Solution Waldc}, and $m_0$ verifies
\[
R \frac{1}{g (-z^*(m))-1}=c
\]
if $m_0>0$, while $m_1$ is given by
\[
\pi(m_1)-L =R \frac{1)}{g (-z^*(m_1))-1} -c
\]
Again, we have that $m^*=\overline{q}$,

\end{proof}
\bigskip

We are going to show that the free boundary is continuous and piecewise differentiable

\begin{lemma} \label{Continuity of z}
The free boundary $z^*(m)$  is continuous in $m$.
\end{lemma}
\begin{proof}[Proof of Lemma \ref{Continuity of z}]
By continuity of the value function in the direction of $z$, we have trivially that $z^*(m)$ is continuous. Nevertheless, let's focus first on $m_0$ 
\[
z^*(\hat{m}-) \equiv  \lim_{m \uparrow \hat{m}} z^*(m) \quad \quad \text{and} \quad \quad z^*(\hat{m}+) \equiv  \lim_{m \downarrow \hat{m}} z^*(m)
\]
Note that, for $m_0$, we have that $m_0+=m_0$ and by continuity of the value function
\begin{equation} \label{Aux continuity at m0}
c   g (z-z^*(m_0-))   = \frac{g (z-z^*(m_0))}{g (-z^*(m_0))-1} R 
\end{equation}
At the same time, recall that at the boundary $z=z^*(m_0)$ it must be that
\[
0= W(0,m_0) - R = \frac{1}{g (-z^*(m_0))-1} R -c 
\]
and therefore, we must have that\footnote{This is the first equation in Corollary \ref{Cutoffs: smooth pasting}.}
\[
c =  \frac{1}{g (-z^*(m_0))-1} R 
\]
and replacing in \eqref{Aux continuity at m0} we get that
\[
  g (z-z^*(m_0-))   = g (z-z^*(m_0)) 
\]
and $z^*(m_0-)=z^*(m_0)$ as desired. Now let's focus on $m_1$ and note that continuity of the value function gives
\begin{equation} \label{Aux continuity at m1}
\frac{g (z-z^*(m_1-)) }{g (-z^*(m_1-))-1} R=(\pi(m_1)-L+c)  g (z-z^*(m_1))  
\end{equation}
At the boundary $z= z^*(m_1)$ we must have that
\[
\pi(m_1)-L = W(0,m_1)-R 
\]
or
\[
\pi(m_1)-L = (\pi(m)-L+c)  g (-z^*(m_1)) -c-R 
\]
which gives\footnote{This is the second equation in Corollary \ref{Cutoffs: smooth pasting}.}
\[
\frac{R}{\pi(m_1)-L+c} = g (-z^*(m_1))-1  
\]
and replacing in \eqref{Aux continuity at m1} we get
\begin{align*}
\frac{g (z-z^*(m_1-)) }{g (-z^*(m_1-))-1} =\frac{ g (z-z^*(m_1))}{g (-z^*(m_1))-1}  
\end{align*}
and since $\frac{g (z-x) }{g (-x)-1} $ is monotonic in $x$ and we get the desired result $z^*(m_1-)=z^*(m_1)$. Finally, if $m^*<\infty$ we must have that
\[
(\pi(m^*)-L+c) ( g (z-z^*(m^*-)) -1)=0
 \]
and since $\pi(m^*)-L+c>0$ it must be that $g (z-z^*(m^*-)) =1$ and therefore $z=z^*(m^*-)$ for any feasible $z$. This is only possible if $z^*(m^*-)=0$ and there is no feasible $z>z^*(m^*-)$. 
\end{proof}

The next lemma shows that $z^*(m)$ is piecewise differentiable by applying the reflection condition
\[
\frac{\partial W(0,m)}{\partial m} = \lim_{z \uparrow 0} \frac{\partial W(z,m)}{\partial z}
\]

\begin{lemma} \label{Lemma Reflwction}
The boundary $z^*(m)$ verifies the differentiable conditions in \eqref{reflection in prop}.
\end{lemma}

\begin{proof}[Proof of Lemma \ref{Lemma Reflwction}]

Consider first \underline{the states $ m \in [m_1, m^*) $} which are the states that verify,
\[
W(0,m)-R < \pi(m)-L \leq W(0,m)
\]
In these states the free boundary determines that the product is launched and no further developments occur, so we have that 
\[
W^*(z^*(m),m) = \pi(m)-L
\]
and \eqref{Solution Waldb} reduces to
\[
W(z,m)  =    \left( \pi(m)-L+c  \right)  g (z-z^*(m))  -c
\]
and verifies $W(z,m)  \geq \pi(m)-L $ using that  $g (z-z^*(m)) \geq 1$.  The reflection condition \eqref{reflection c} reduces to
\begin{align} \label{Reflection 2}
% \gamma_1    \frac{\partial z^*(m)}{\partial m}   & =   \left( \frac{\pi'(m)}{\pi(m)-L+c}- \gamma_1 \right) - \frac{\gamma_1 -\gamma_2}{ \gamma_2} \frac{e^{-\gamma_2 z(m) }}{ e^{-\gamma_1 z(m) }-e^{-\gamma_2 z(m) } }  \frac{\pi'(m)}{\pi(m)-L+c}. \\
  1+\frac{\partial z^*(m)}{\partial m}  & =   \frac{g(-  z(m) )} {g'(-  z(m) )}  \frac{\pi'(m)}{\pi(m)-L+c} >0
 \end{align}

Existence of $z^*(m)$ follows as an application of 
\begin{theorem}[Picard-Lindelof] \label{Theorem: Picard-Lindelof}
Consider an initial value problem 
\begin{equation} \label{Differential equation}
y'(t) = f(x,y), \quad y(x_0)=y_0
\end{equation}
If $f(x,y)$ and $\frac{\partial f(x,y)}{\partial y}$ are continuous in some open interval $R \equiv \left\{ (a,b) \times (c,d)\right\}$ with $(x_0,y_0) \in R$, then the ordinary differential equation \eqref{Differential equation} has a unique solution in some close interval $[x_0-h,x_0+h]$ where $h>0$
\end{theorem}
Unfortunately, the explicit solution to the differential equation \eqref{Reflection 2} is not known as this is another representation of an Abel equation. Let $y(m)= e^{-(\gamma_1-\gamma_2) z^*(m) }$  so we have that \eqref{Reflection 2} is equivalent to
\begin{align*}
% y'(m) (y(m)-1) & =  (\gamma_1-\gamma_2) y(m) (y(m)-1) + \frac{\gamma_1-\gamma_2}{\gamma_2 \gamma_1} y(m) (\gamma_1 -\gamma_2 y(m))  \frac{\pi'(m)}{\pi(m)-L+c} \\
   y'(m) \frac{y(m)-1}{y(m) } & =  \frac{\gamma_1-\gamma_2}{\gamma_1}  \left[    \frac{\gamma_1-\gamma_2}{\gamma_2} \frac{\pi'(m)}{\pi(m)-L+c} 
  +  (y(m)-1) \left( \gamma_1  - \frac{\pi'(m)}{\pi(m)-L+c} \right) \right]
 \end{align*}
 and rearranging we get
\begin{align} \label{equation:Abel 2}
 y'(m) (y(m)-1) & =  f_1(m)  y(m) + f_2(m) (y(m) )^2  
 \end{align}
  for %$g(m)  \equiv  -1$ and 
\begin{align*}
f_1(m) & \equiv \frac{\gamma_1-\gamma_2}{\gamma_2} \left( \frac{\pi'(m)}{\pi(m)-L+c}-  \gamma_2\right) \\
f_2(m) & \equiv  \frac{\gamma_1-\gamma_2}{\gamma_1} \left( \gamma_1-\frac{\pi'(m)}{\pi(m)-L+c}\right)
\end{align*}
which is an Abel equation of the second kind (see \cite{book:handbookODE} page 33). Using the substitution
\[
w(x) = (y(x)-1)e^{\int_{x}^{\overline{m}} f_2(z) dz}
\]
in \eqref{equation:Abel 2} we get the canonical representation of an Abel equation 
\begin{align} \label{equation:Abel 2 can}
w'(x)   w(x)  & = F_1(x)  w(x)  +F_0(x)
 \end{align}
 for
 \begin{align*}
F_0(x) & \equiv ( f_1(x)+f_2(x) ) e^{ 2\int_{x}^{\overline{m}} f_2(z) dz} \\
F_1(x)  & \equiv  ( f_1(x)+2 f_2(x) )  e^{  \int_{x}^{\overline{m}} f_2(z) dz} 
 \end{align*}
Alternatively by setting 
\[
u(x) = (y(x)-1)^{-1} 
\]
in \eqref{equation:Abel 2} we get and Abel equation of the first kind
\begin{align} \label{equation: Abel 1}
\frac{u'(x)}{u(x)  (1+u(x))}   =    f_1(x)  u(x)   +  f_2(x)  (1+u(x)) 
\end{align}

\medskip

Now let's focus on the states $[m_0,m_1)$ which are the states that verify,
\[
0 < \pi(m)-L \leq W(0,m)-R
\]
In these states the free boundary determines that the development process is restarted from the highest available quality, 
\[
W^*(z^*(m),m) = W^*(0,m) - R
\]
and \eqref{Solution Waldb} reduces to
\[
W(z,m)  =   R \frac{g (z-z^*(m))}{g (-z^*(m))-1} - c
\]
and therefore \eqref{reflection c} reduces to
\begin{align} \label{Reflection 2b}
 1+\frac{\partial z^*(m)}{\partial m}   =g (-z^*(m)) % \quad \Rightarrow \quad \frac{\partial^2 z^*(m)}{\partial m^2}   =-g' (-z^*(m)) \frac{\partial z^*(m)}{\partial m} 
 \end{align}
 which we know it has a solution by the Picard-Lindelof Theorem. Using the expression for $g$, letting $h(m) \equiv m + z^*(m)$, and rearranging
\[
 \frac{\partial h(m)}{\partial m}   =\frac{\gamma_1}{\gamma_1-\gamma_2} e^{\gamma_2 m } e^{-\gamma_2 h(m)} -\frac{\gamma_2}{\gamma_1-\gamma_2} e^{\gamma_1 m  } e^{-\gamma_1 h(m) }
\]
Let 
 \[
 F(m) = e^{\gamma_1 h(m)} \quad  \Rightarrow \quad  F'(m) = \gamma_1 h'(m) F(m) 
 \]
 so replacing we have that the differential equation \eqref{Reflection 2b} is equal to
 \begin{align} \label{Bernoulli}
% \frac{\gamma_1}{\gamma_1-\gamma_2} \gamma_ 1  e^{\gamma_1 m}  & = F'(x) +\frac{\gamma_1}{\gamma_1-\gamma_2} \gamma_ 2  e^{\gamma_2 m} \left(F(x)\right)^{\frac{(\gamma_1-\gamma_2)}{\gamma_1}} \\
\alpha F'(m)  + \gamma_ 1 e^{\gamma_1 m}      &=   \gamma_ 2 e^{\gamma_2 m}  \left(F(m)\right)^{\alpha} 
 \end{align}
 for $\alpha \equiv \frac{\gamma_1-\gamma_2}{\gamma_1 }$. This resembles a Bernoulli differential equation,
\[
\frac{dy}{dx}  + P(x) y(x)=  f(x) (y(x))^{\alpha}  
\]
with a major difference as the linear term is a constant. In fact, it is a particular form of a generalized Bernoulli equation\footnote{The original Bernoulli equation, according to \cite{azevedo2017}, is given by
which differs from \eqref{Bernoulli} in just the term on the right hand side.
} 
 \[
 y'(x) + P(x) h(y(x)) = f(x) g(y(x)) 
 \]
for 
\begin{gather*}
P(\sigma)  =  \frac{ \gamma_ 1 e^{\gamma_1 \sigma}   }{\alpha} \quad  \text{and}  \quad f(\sigma)  =  \frac{\gamma_ 2 e^{\gamma_2 \sigma}}{\alpha}  \\
h(\sigma)  = 1 \quad  \text{and}  \quad g(\sigma)  = \sigma^\alpha 
\end{gather*}
which is not among the equations solvable with \cite{azevedo2017}.

Finally, consider, $m <m_0$ and recall that 
\[
W^*(z(m),m) =0 
\]
which implies that 
\[
W^*(z,m) = c   (g (z-z^*(m))-1)
\]
and therefore \eqref{reflection c} reduces to
\begin{align} \label{Reflection 2c}
 \frac{\partial z^*(m)}{\partial m} =-1
 \end{align}
which gives the explicit result $z^*(m) +m = z^*(m_0) + m_0 $ and therefore $z^*(0) = z^*(m_0) + m_0 $.

\medskip

It remains to show that the optimal program is feasible which requires that
\[
W(z^*(m),m) \geq \max\{0,\pi(m)-L\}
\]
and given the value function, this is equivalent to $W(z^*(m),m) >0$ or using \eqref{Solution Waldc} the value function reduces to
\begin{align*}
W(z^*(m),m) & =  
\begin{cases}
\frac{R}{g (-z^*(m))-1}  - c  & \text{ if } m \in [m_0, m_1) \\
\pi(m)-L  & \text{ if } m \geq m_1 
\end{cases}
\end{align*}
and since $W(z^*(m),m)$ is increasing in $m$ it is necessary and sufficient if
\begin{align*}
c & \leq \frac{R}{g (-z^*(m_0))-1}  \\
L & \leq \pi(m_1)   
\end{align*}
and since the first one holds by definition, while the definition of $m_1$ implies that the second one is equivalent to
\begin{align*}
g(-z^*(m_0)) = \frac{R+c}{c } \geq g(-z^*(m_1)) 
\end{align*}
which is true since $z^*(m_1)>z^*(m_0)$. To see the order of the cutoffs, note that \eqref{Reflection 2} gives
\[
z^*(m_0)+m_0=z^*(m)+m, \; \forall m \in [0,m_0)
\]
and \eqref{Reflection 2b} gives
\[
z^*(m_1) = z^*(m_0)+ \int_{m_0}^{m_1} \left(g (-z^*(x))-1 \right) dx, \; \forall m \in [m_0,m_1)
 \]

\end{proof}

\medskip

It remains to check that $W$ is continuously differentiable in $z$ twice and in $m$ once. It is immediate to see that is differentiable in $z$ even at the boundary since we are using $W(z,m) = W(\underline{z}^*(m),m)$ at all $z<\underline{z}^*(m)$ so we focus on the differentiability of $W(z,m)$ in $m$. This is not trivial as $z^*(m)$ is not differentiable at the cutoffs. Nevertheless, we have
\begin{itemize}
\item Consider the cutoff $m_0$ and note that from the left, we have that
\begin{align*}
\frac{\partial W(0,m)}{\partial m} & = \frac{ \partial \left( c (g(- z^*(m))-1) \right) }{\partial m} =- c g'(- z^*(m))  \frac{\partial z^*(m)}{\partial m} = c g'(- z^*(m))
\end{align*}
and from the right
\begin{align*}
\frac{\partial W(0,m)}{\partial m}  & = \frac{ \partial \left( R \frac{g(- z^*(m))}{g(- z^*(m))-1} -c \right) }{\partial m} \\
& =     \frac{R}{g(- z^*(m))-1} \frac{g'(- z^*(m))}{(g(- z^*(m))-1)} \frac{ \partial z^*(m)}{\partial m} \\
& =     \frac{R}{g(- z^*(m))-1} g'(- z^*(m)) 
\end{align*}
Note that taking limits, we have
\begin{align*}
\lim_{m \uparrow m_0} \frac{\partial W(0,m)}{\partial m} & =c g'(- z^*(m_0)) \\
\lim_{m \downarrow m_0} \frac{\partial W(0,m)}{\partial m} & =   \frac{R}{g(- z^*(m_0))-1} g'(- z^*(m_0))=   c g'(- z^*(m_0))
\end{align*}
\item Consider now $m_1$ and note that from the left, we have that
\begin{align*}
\frac{\partial W(0,m)}{\partial m} & =  \frac{R}{g(- z^*(m))-1} \frac{g'(- z^*(m))}{(g(- z^*(m))-1)} \frac{ \partial z^*(m)}{\partial m} \\
& =  \frac{R}{g(- z^*(m))-1} g'(- z^*(m))
\end{align*}
and from the right we have
\begin{align*}
\frac{\partial W(0,m)}{\partial m} & =  \frac{\partial \left( \left( \pi(m)-L+c\right) g(-z^*(m))-c\right)}{\partial m} \\
& = \pi'(m) g(-z^*(m)) - \left( \pi(m)-L+c\right) g'(-z^*(m)) \frac{\partial z^*(m)}{\partial m} \\
& =  \left( \pi(m)-L+c\right) g'(-z^*(m)) 
\end{align*}
and taking limits from the right and from the right we have
\begin{align*}
\lim_{m \uparrow m_1} \frac{\partial W(0,m)}{\partial m} &  =  \frac{R}{g(- z^*(m_1))-1} g'(- z^*(m_1)) = \lim_{m \downarrow m_1}  \frac{\partial W(0,m)}{\partial m} 
\end{align*}
when we used the definition of the cutoff $m_1$.
\end{itemize}

\medskip

We finish the proof by providing some further characterization results

\textbf{The case of $m \in [m_1,m^*)$:} Having no explicit solution for $z^*(m)$ when $m \in [m_1,m^*)$, we focus on understanding the behavior of the free boundary as an implicit solution. First, note that the total differential of \eqref{Reflection 2}  is given by
\begin{align*}
%\left( 1+\frac{\partial z^*(m)}{\partial m}\right)  \frac{g'(-  z(m) )}{g(-  z(m) )}  & =    \frac{\pi'(m)}{\pi(m)-L+c} 
\frac{\partial^2 z^*(m)}{\partial m^2}   &
+\left(     \frac{ g'(-  z(m) )}{g(-  z(m) )  } - \frac{g''(-  z(m) )} {g'(-  z(m) ) } \right)   \frac{\partial z^*(m)}{\partial m} \left( 1+\frac{\partial z^*(m)}{\partial m} \right) \\
& =  \left( \frac{\pi''(m)}{\pi'(m)} -   \frac{\pi'(m)}{\pi(m)-L+c}  \right) \left( 1+\frac{\partial z^*(m)}{\partial m} \right) <0
\end{align*}
Recalling the expression for $g$ we have
\[
g'(x) =  \gamma_2 \left(  g(x) - e^{\gamma_1 x} \right) \quad \quad g''(x) =  \gamma_2 \left(  g'(x) - \gamma_1 e^{\gamma_1 x} \right)
\]
and therefore we 
\[
\frac{g'(x)}{g(x)} -\frac{g''(x)}{g'(x)}   =      \frac{\gamma_2  e^{\gamma_1 x}}{g(x)}   \frac{\gamma_1 e^{\gamma_2 x}}{g'(x)} <0
\]
which gives
\begin{align*}
%\left( 1+\frac{\partial z^*(m)}{\partial m}\right)  \frac{g'(-  z(m) )}{g(-  z(m) )}  & =    \frac{\pi'(m)}{\pi(m)-L+c} 
\frac{\frac{\partial^2 z^*(m)}{\partial m^2} }{1+\frac{\partial z^*(m)}{\partial m} }  &
+  \frac{\gamma_2  e^{-\gamma_1 z^*(m)}}{g(-z^*(m))}   \frac{\gamma_1 e^{-\gamma_2 z^*(m)}}{g'(-z^*(m))}   \frac{\partial z^*(m)}{\partial m}  =   \frac{\pi''(m)}{\pi'(m)} -   \frac{\pi'(m)}{\pi(m)-L+c}  <0
\end{align*}
It follows that, if $ -1< \frac{\partial z^*(m)}{\partial m} \leq 0$, then $\frac{\partial^2 z^*(m)}{\partial m^2}<0$ and $z^*(m)$ is decreasing.

We have then three possible cases

\begin{enumerate}
\item If $\frac{\partial z^*(m)}{\partial m}_{m=m_1} < 0$, then $z^*(m_1)$ decreases continuously until $m$ reaches $\overline{q}$ and it reaches it at some $z^*(\overline{q})<0$. Note that
\begin{align*}
 \lim_{m \uparrow \overline{q}} \frac{\partial z^*(m)}{\partial m}  & =   \lim_{m \uparrow \overline{q}}\left[  \frac{g(-  z(m) )} {g'(-  z(m) )}  \frac{\pi'(m)}{\pi(m)-L+c} -1 \right] = -1
  \end{align*}
  Moreover,  we must have some $\delta>0$ such that $z^*(m)<-\delta$ and therefore
\[
W(z,m) \geq (\pi(m)-L+c) g(z-\delta) - c
\]
which implies that
\[
W(0,\overline{q}) \geq (\pi(\overline{q})-L+c) g(-\delta) - c > \pi(\overline{q})-L
\]
which is not possible.

\item If $\frac{\partial z^*(m)}{\partial m} \geq 0$ for all $m>m_1$. There are therefore two situations: $\lim_{m \uparrow \overline{q}} z^*(m)<0$ or $\lim_{m \uparrow \overline{q}} z^*(m) = 0$

\begin{itemize}
\item If $\lim_{m \uparrow \overline{q}} z^*(m)<0$,  it is trivial to see that
\[
\lim_{m \uparrow \overline{q}}  \frac{\partial z^*(m)}{\partial m} = 0
\]
because $\pi'(\overline{q})=0$.  Again, we must have some $\delta>0$ such that $z^*(m)<-\delta$ and therefore
\[
W(z,m) \geq (\pi(m)-L+c) g(z-\delta) - c
\]
which implies that
\[
W(0,\overline{q}) \geq (\pi(\overline{q})-L+c) g(-\delta) - c > \pi(\overline{q})-L
\]
which is not possible. 
\item It follows then we must have that $\lim_{m \uparrow \overline{q}} z^*(m) = 0$, where, we need to apply L'Hopital's rule to get
\begin{align*}
\lim_{m \uparrow \overline{q}} \left( \frac{\partial z^*(m)}{\partial m}  \right) & = \lim_{m \rightarrow \overline{q}}  \frac{\pi'(m)}{\pi(m)-L+c} \frac{g(-z^*(m))}{g'(-z^*(m))} -1\\
& = \frac{1}{\pi(\overline{q})-L+c} \lim_{m \rightarrow \overline{q}}   \frac{\pi''(m)}{-g''(-z^*(m))\frac{\partial z^*(m)}{\partial m}} -1 \\
& = \frac{1}{\pi(\overline{q})-L+c}   \frac{\pi''(\overline{q})}{\gamma_1 \gamma_ 2  \lim_{m \uparrow \overline{q}} \left( \frac{\partial z^*(m)}{\partial m}  \right)} - 1
\end{align*}
and therefore
\begin{align*}
\lim_{m \uparrow \overline{q}} \left( \frac{\partial z^*(m)}{\partial m}  \right)+ \left( \lim_{m \uparrow \overline{q}} \left( \frac{\partial z^*(m)}{\partial m}  \right) \right)^2 & =  \frac{1}{\pi(\overline{q})-L+c}   \frac{\pi''(\overline{q})}{\gamma_1 \gamma_ 2  } \\
 \lim_{m \uparrow \overline{q}} \left( \frac{\partial z^*(m)}{\partial m}  \right)  & =  \sqrt{\frac{1}{\pi(\overline{q})-L+c}   \frac{\pi''(\overline{q})}{\gamma_1 \gamma_ 2  } +\frac{1}{4}} -\frac{1}{2} >0
\end{align*}
\end{itemize}

\item If $\frac{\partial z^*(m)}{\partial m}_{m=m_1} > 0$  and there is some $\overline{m}$ such that $ \frac{\partial z^*(m)}{\partial m}_{m=\overline{m}} = 0$. In this case we have that 
\[
 \frac{\pi'(\overline{m})}{\pi(\overline{m})-L+c} =  \frac{g(-  z(\overline{m}) )} {g'(-  z(\overline{m}) )} 
 \]
Recalling that if $ -1< \frac{\partial z^*(m)}{\partial m} \leq 0$, then $\frac{\partial^2 z^*(m)}{\partial m^2}<0$, it follows that $\frac{\partial z^*(m)}{\partial m} < 0$ if $ m > \overline{m}$ and we are back in the first case

\end{enumerate}

It follows that $ \frac{\partial z^*(m)}{\partial m}>0$ for all $m \in [m_1,m^*)$ and
\begin{align*}
 \lim_{m \uparrow \overline{q}} \left( \frac{\partial z^*(m)}{\partial m}  \right)  & =  \sqrt{\frac{1}{\pi(\overline{q})-L+c}   \frac{\pi''(\overline{q})}{\gamma_1 \gamma_ 2  } +\frac{1}{4}} -\frac{1}{2} >0
\end{align*}

\medskip

\textbf{The case of $m \in [m_0,m_1)$:} Again, having no explicit solution for $z^*(m)$ when $m \in [m_0,m_1)$, we focus on understanding the behavior of the free boundary as an implicit solution. 
Differntiating \eqref{Reflection 2b} we get
\begin{align*} 
\frac{\partial^2 z^*(m)}{\partial m^2}  & =-g'(-  z^*(m) ) \frac{\partial z^*(m)}{\partial m} 
 \end{align*}
 which implies that
\begin{align*} 
sign \left\{ \frac{\partial^2 z^*(m)}{\partial m^2}  \right\} & =- sign \left\{ \frac{\partial z^*(m)}{\partial m} \right\}
 \end{align*}
Using
\[
g (-z^*(m_0))  = \frac{R  + c}{c} >1
\]
from \eqref{Cutoffs: smooth pasting} in the reflection condition \eqref{Reflection 2b} we get that 
\[
\lim_{m \downarrow m_0} \frac{\partial z^*(m)}{\partial m}   =\frac{R  }{c} >0
\]
On the other hand, using from \eqref{Cutoffs: smooth pasting} gives that
\[
g (-z^*(m_1)) = 1+ \frac{R}{\pi(m_1)-L+c} 
\]
in the \eqref{Reflection 2b} we get that 
\[
\lim_{m \uparrow m_1} \frac{\partial z^*(m)}{\partial m}   =\frac{R}{\pi(m_1)-L+c} <\frac{R}{c} 
\]

\newpage

To determine whether $m_0>0$ or not we note that the system that determines the equilibrium is given by
\begin{align} \label{system to solve}
0 & = z^*(\overline{q})  && \\
\overline{q} & = z^*(m)+ m+ \int_m^{\overline{q}}  \frac{\pi'(x)}{\pi(x)-L+c} \frac{g(-z^*(x))}{g'(-z^*(x))} dx  & & m \in (m_1,\overline{q}] \nonumber\\
c & =  \frac{R}{g(-z^*(m_1))-1} - (\pi(m_1)-L) & &  \nonumber\\
z^*(m_1)+ m_1  & =z^*(m)+ m+ \int_{m}^{m_1} g(-z^*(x)) dx  & & m \in (m_0,m_1) \nonumber\\
c & =  \frac{R}{g(-z^*(m_0))-1}  & &  \nonumber  \\
z^*(m_0) + m_0 & = z^*(m) + m    & & m \in [0,m_0) \nonumber 
\end{align}
Note that for $m_0 >0$ we must have that 
\[
z^*(m_1)+ m_1  < z^*(m_0)+ \int_{0}^{m_1} g(-z^*(x)) dx 
\]
so it must be that 
\begin{align*} 
0 & <z^*(m_0)+ \int_{0}^{m_1} g(-z^*(x)) dx+ \int_{m_1}^{\overline{q}}  \frac{\pi'(x)}{\pi(x)-L+c} \frac{g(-z^*(x))}{g'(-z^*(x))} dx  -\overline{q}
\end{align*}
Let
\[
G_{\overline{q}}(y) = z^*(y)+y+ \int_{y}^{m_1} g(-z^*(x)) dx+ \int_{m_1}^{\overline{q}}  \frac{\pi'(x)}{\pi(x)-L+c} \frac{g(-z^*(x))}{g'(-z^*(x))} dx  -\overline{q}
\]
and note that if $G_{\overline{q}}(0) <(>)0$, then $m_0>(=)0$. Note that 
\begin{align*}
G_{\overline{q}}(0) & = z^*(0)+ \int_{0}^{m_1} g(-z^*(x)) dx+ \int_{m_1}^{\overline{q}}  \frac{\pi'(x)}{\pi(x)-L+c} \frac{g(-z^*(x))}{g'(-z^*(x))} dx  -\overline{q} \\
& =z^*(m_1)+ m_1+ \int_{m_1}^{\overline{q}}  \frac{\pi'(x)}{\pi(x)-L+c} \frac{g(-z^*(x))}{g'(-z^*(x))} dx  -\overline{q}
\end{align*}
we want to show that for $\overline{q}$ low $G_{\overline{q}}(0)<0$ and for $\overline{q}$ high $G_{\overline{q}}(0)>0$. The next claim shows how $\overline{q}$ affects the different equilibrium cutoffs.

\qed

\newpage

% \end{proof}

\subsection*{Proof of Lemma \ref{Lemma comparative results c}}

%\begin{proof}[Proof of Lemma \ref{Claim comparative results c}]

Let's focus first on $m >m_1$ where we have that the first line of \eqref{system to solve} gives
\begin{align} \label{auxiliar for q}
1 & = \frac{ \partial z^*(m)}{\partial \overline{q}} +  \frac{\pi'(\overline{q})}{\pi(\overline{q})-L+c} \frac{g(-z^*(\overline{q}))}{g'(-z^*(\overline{q}))} - \mathcal{H}(m)  
\end{align}
where
\[
\mathcal{H}(m) \equiv \int_m^{\overline{q}}  h(x) \frac{ \partial z^*(x)}{\partial \overline{q}} dx
\]
and 
\[
h(m) \equiv \frac{\pi'(m)}{\pi(m)-L+c}  \frac{g(-z^*(m))}{g'(-z^*(m))} \left( \frac{g'(-z^*(m))}{g(-z^*(m))}-\frac{g''(-z^*(m))}{g'(-z^*(m))} \right) <0
\]
Using the function $\mathcal{H}(m)$ and the fact that
\begin{align*}
 \frac{\pi'(\overline{q})}{\pi(\overline{q})-L+c} \frac{g(-z^*(\overline{q}))}{g'(-z^*(\overline{q}))} =&  \lim_{m \uparrow q} \left( \frac{\partial z^*(m)}{\partial m}  \right) +1 \\
  & =  \sqrt{\frac{1}{\pi(\overline{q})-L+c}   \frac{\pi''(\overline{q})}{\gamma_1 \gamma_ 2  } +\frac{1}{4}} +\frac{1}{2} >0
\end{align*}
we get
\begin{align*}
 \mathcal{H}'(m)+  h(m) \mathcal{H}(m)  & =  h(m)  \lim_{m \uparrow \overline{q}} \left( \frac{\partial z^*(m)}{\partial m}  \right)  \\
\mathcal{H}'(m)+  h(m) \mathcal{H}(m)  e ^{\int_m^{\overline{q}} h(x)dx} & = h(m) e ^{\int_m^{\overline{q}} h(x)dx}  \lim_{m \uparrow \overline{q}} \left( \frac{\partial z^*(m)}{\partial m}  \right) \\
\frac{d \left(\mathcal{H}(m) e^{-\int_m^{\overline{q}} h(x) dx } \right)}{d m} & =  \frac{d \left( e^{-\int_m^{\overline{q}} h(x) dx } \right) }{dm}   \lim_{m \uparrow \overline{q}} \left( \frac{\partial z^*(m)}{\partial m}  \right) \\
\mathcal{H}(m)   & =  \left( 1-e^{\int_m^{\overline{q}} h(x) dx } \right)   \lim_{m \uparrow \overline{q}} \left( \frac{\partial z^*(m)}{\partial m} \right) >0
 \end{align*}
where we used that $\mathcal{H}(\overline{q})=0$. Using the definition of $ \mathcal{H}(m)$ we get
\begin{align*}
 - h(m) \frac{ \partial z^*(m)}{\partial \overline{q}} & = \mathcal{H}'(m)   \\
\frac{ \partial z^*(m)}{\partial \overline{q}} & =  -  e^{\int_m^{\overline{q}} h(x) dx } \lim_{m \uparrow \overline{q}} \left( \frac{\partial z^*(m)}{\partial m}  \right) < 0
\end{align*}
for all $m >m_1$. In particular
\begin{align}\label{auxiliar for q b}
\frac{ \partial z^*(\overline{q})}{\partial \overline{q}} & = -  \lim_{m \uparrow \overline{q}} \left( \frac{\partial z^*(m)}{\partial m}  \right) 
\end{align}

Now consider the system to determine $(z^*(m_1),m_1)$ which is given by
\begin{align*}
\overline{q} & = z^*(m_1)+ m_1+ \int_{m_1}^{\overline{q}}  \frac{\pi'(x)}{\pi(x)-L+c} \frac{g(-z^*(x))}{g'(-z^*(x))} dx \nonumber\\
\pi(m_1)-L+ c & =  \frac{R}{g(-z^*(m_1))-1}  
\end{align*}
and differentiating we get
\begin{align*}
1 - \frac{ d z^*(m_1)}{d \overline{q}} + \mathcal{H}(m_1) & =    \frac{\pi'(\overline{q})}{\pi(\overline{q})-L+c} \frac{g(-z^*(\overline{q}))}{g'(-z^*(\overline{q}))} -  \left( \frac{\pi'(m_1)}{\pi(m_1)-L+c} \frac{g(-z^*(m_1))}{g'(-z^*(m_1))} -1 \right) \frac{ d m_1}{d \overline{q}} \\
\frac{\pi'(m_1)}{\pi(m_1)-L+ c} \frac{ d m_1}{d \overline{q}} & =  \frac{g'(-z^*(m_1))}{g(-z^*(m_1))-1}   \frac{ d z^*(m_1)}{d \overline{q}}
\end{align*}
Using \eqref{auxiliar for q} and \eqref{auxiliar for q b} we get
\begin{align*}
 e^{\int_{m_1}^{\overline{q}} h(x) dx }   \lim_{m \uparrow \overline{q}} \left( \frac{\partial z^*(m)}{\partial m} \right) & =    \left( \frac{\pi'(m_1)}{\pi(m_1)-L+c} \frac{g(-z^*(m_1))}{g'(-z^*(m_1))} -1 \right) \frac{ d m_1}{d \overline{q}} -  \frac{ d z^*(m_1)}{d \overline{q}}\\
\frac{\pi'(m_1)}{\pi(m_1)-L+ c} \frac{ d m_1}{d \overline{q}} & =  \frac{g'(-z^*(m_1))}{g(-z^*(m_1))-1}   \frac{ d z^*(m_1)}{d \overline{q}}
\end{align*}
and therefore we get \eqref{effects of q}. Unfortunately, the effects on $m_1$ and $z^*(m_1)$ are not determined, but 
\begin{align*}
 \frac{ d m_1}{d \overline{q}}         \frac{\pi'(m_1)}{\pi(m_1)-L+c}  & =  \frac{ d z^*(m_1)}{d \overline{q}} \frac{g'(-z^*(m_1))}{g(-z^*(m_1))-1} 
 \end{align*}
so $\frac{ d m_1}{d \overline{q}}  \frac{ d z^*(m_1)}{d \overline{q}} >0$ and both move in the same direction.

Now consider $m \in (m_0,m_1)$ and \eqref{system to solve} gives 
\begin{equation*} 
\frac{d z^*(m_1)}{d \overline{q}}  -(g(-z^*(m_1))-1)  \frac{d m_1}{d \overline{q}}  =\frac{\partial z^*(m)}{\partial \overline{q}} - \int_{m}^{m_1} g'(-z^*(x)) \frac{\partial z^*(x)}{\partial \overline{q}} dx  
\end{equation*}
or using the previous result for $ \frac{d m_1}{d \overline{q}} $ we get
\begin{equation*}
 (g(-z^*(m_1))-1)    e^{\int_{m_1}^{\overline{q}} h(x) dx }  \left( \lim_{m \uparrow \overline{q}} \left( \frac{\partial z^*(m)}{\partial m}  \right)\right)=\frac{\partial z^*(m)}{\partial \overline{q}} - \int_{m}^{m_1} g'(-z^*(x)) \frac{\partial z^*(x)}{\partial \overline{q}} dx  
\end{equation*}
Let $\mathcal{M}(m) \equiv \int_{m}^{m_1} g'(-z^*(x)) \frac{\partial z^*(x)}{\partial \overline{q}} dx$ to get
\begin{equation*} 
-g'(-z^*(m)) e^{\int_{m_1}^{\overline{q}} h(x) dx }  (g(-z^*(m_1))-1)  \left(  \lim_{m \uparrow \overline{q}} \left( \frac{\partial z^*(m)}{\partial m}  \right)  \right)  =\mathcal{M}'(m) + g'(-z^*(m)) \mathcal{M}(m)  
\end{equation*}
with solution given by
\begin{equation*} 
  \mathcal{M}(m)  =  \left( e^{\int^{m_1}_m g'(-z^*(x)) dx } -1 \right)e^{\int_{m_1}^{\overline{q}} h(x) dx }  (g(-z^*(m_1))-1)  \left(  \lim_{m \uparrow \overline{q}} \left( \frac{\partial z^*(m)}{\partial m}  \right)  \right)  >0
\end{equation*}
and therefore for all $m \in (m_0,m_1)$
\[
\frac{\partial z^*(m)}{\partial \overline{q}} = e^{\int^{m_1}_m g'(-z^*(x)) dx }  e^{\int_{m_1}^{\overline{q}} h(x) dx }  (g(-z^*(m_1))-1)  \left(  \lim_{m \uparrow \overline{q}} \left( \frac{\partial z^*(m)}{\partial m}  \right)  \right)  >0
\]
Now note that for all $m \in (m_0,m_1)$ 
\begin{align*}
\int_m^{m_1} g'(-z^*(x)) dx & = \int_m^{m_1} \frac{g'(-z^*(x))}{g(-z^*(x))-1} (g(-z^*(x))-1) dx \\
& = -\int_m^{m_1} \frac{\partial \ln \left( g(-z^*(x))-1 \right)}{\partial z^*(x)} \frac{\partial z^*(x)}{\partial x} dx \\
& = - \left(  \ln \left( g(-z^*(x))-1 \right) \right)_m^{m_1} \\
& =  \ln \left( \frac{ g(-z^*(m))-1}{g(-z^*(m_1))-1} \right) 
\end{align*}
and therefore
\[
e^{\int_m^{m_1} g'(-z^*(x)) dx} = e^{\ln \left( \frac{ g(-z^*(m))-1}{g(-z^*(m_1))-1} \right)} = \frac{ g(-z^*(m))-1}{g(-z^*(m_1))-1}
\]
which implies that
\begin{align*}
\frac{\partial z^*(m)}{\partial \overline{q}} & = e^{\int_{m_1}^{\overline{q}} h(x) dx }  (g(-z^*(m))-1)  \left(  \lim_{m \uparrow \overline{q}} \left( \frac{\partial z^*(m)}{\partial m}  \right)  \right) > 0  \\
 \mathcal{M}(m)  & =  \left(  g(-z^*(m))-g(-z^*(m_1)) \right)e^{\int_{m_1}^{\overline{q}} h(x) dx }   \left(  \lim_{m \uparrow \overline{q}} \left( \frac{\partial z^*(m)}{\partial m}  \right)  \right) >0 
\end{align*}

Consider now the state $\left( z^*(m_0), m_0 \right)$ and noting that $\frac{d z^*(m_0)}{d \overline{q}} =0$ we have
\begin{align*} 
(g(-z^*(m_0))-1)\frac{d m_0}{d \overline{q}}   & =- \int_{m_0}^{m_1} g'(-z^*(x)) \frac{\partial z^*(x)}{\partial \overline{q}} dx - \left( \frac{d z^*(m_1)}{d \overline{q}}  -(g(-z^*(m_1))-1)  \frac{d m_1}{d \overline{q}} \right) \\
\frac{R}{c} \frac{d m_0}{d \overline{q}} & =- \mathcal{M}(m_0)  - \left( \frac{d z^*(m_1)}{d \overline{q}}  -(g(-z^*(m_1))-1)  \frac{d m_1}{d \overline{q}} \right) \\
 \frac{d m_0}{d \overline{q}} & =- e^{\int_{m_1}^{\overline{q}} h(x) dx }   \left(  \lim_{m \uparrow \overline{q}} \left( \frac{\partial z^*(m)}{\partial m}  \right)  \right) <0
\end{align*}
 
It is easy to see that for $m<m_0$ we have that
\[
\frac{dm_0}{d \overline{q}}  = \frac{dz^*(m)}{d \overline{q}}
\]

%\end{proof}

\qed

Note that
\begin{align*}
\frac{d G_{\overline{q}}(0)}{d \overline{q}} & =\frac{d z^*(m_1)}{d \overline{q}}  +  \frac{\pi'(\overline{q})}{\pi(\overline{q})-L+c} \frac{g(-z^*(\overline{q}))}{g'(-z^*(\overline{q}))}   - \left( \frac{\pi'(m_1)}{\pi(m_1)-L+c} \frac{g(-z^*(m_1))}{g'(-z^*(m_1))} -1 \right)  \frac{d m_1}{d \overline{q}}  -1
\end{align*}
and using \eqref{effects of q} we get
\begin{align*}
\frac{d G_{\overline{q}}(0)}{d \overline{q}} & = \left( 1 - e^{\int_{m_1}^{\overline{q}} h(x) dx } \right) \lim_{m \uparrow \overline{q}} \left( \frac{\partial z^*(m)}{\partial m}  \right)   >0 
\end{align*}

\newpage

\subsection*{Proof of Lemma \ref{Prop: lemma viable}.}

%\begin{proof}[Proof of Lemma \ref{Prop: lemma viable} ]

It follows by defining the stopping times
\begin{align*}
\overline{\mathbb{T}}  \equiv \inf \left\{t>0: X(t) >m_0 \right\} 
\quad \quad \quad  
\underline{\mathbb{T}}  \equiv \inf \left\{t>0: X(t) <z^*(0) \right\}
\end{align*}
\normalsize
and noting that
\small
\begin{align*}
\Pr \left(M(t) > m_0, \min_{s \leq t} \{ Z(s) - Z^*(M(s)) \} >0 \right) & = \Pr (\overline{\mathbb{T}} > \underline{\mathbb{T}}) \\
& = \Pr (\mathbb{T} = \overline{\mathbb{T}} ) \\
& = \frac{1- e^{ -(\gamma_1+\gamma_2)  z^*(0)}}{e^{- (\gamma_1+\gamma_2)  m_0}- e^{ -(\gamma_1+\gamma_2)  z^*(0)}} \\
& = \frac{e^{ (\gamma_1+\gamma_2)  z^*(0)}- 1}{e^{- (\gamma_1+\gamma_2)  (m_0- z^*(0))}- 1} \\
& = \frac{e^{ (\gamma_1+\gamma_2)  z^*(0)}- 1}{e^{ (\gamma_1+\gamma_2)  z^*(m_0)}- 1} 
\end{align*}
\normalsize
where $\mathbb{T} \equiv \inf \left\{ \underline{\mathbb{T}},\overline{\mathbb{T}}\right\}$. The expression follows by noting $
z^*(m_0) + m_0 = z^*(0)$ and $\gamma_1+\gamma_2 = -\frac{2 \mu}{\sigma^2}$. It follows then that
\small
\begin{align*}
& \frac{ d \Pr \left(M(t) > m_0, \min_{s \leq t} \{ Z(s) - Z^*(M(s)) \} >0 \right) }{dR(L)} \\
&  = (\gamma_1+\gamma_2)  e^{ (\gamma_1+\gamma_2)  (z^*(0)+ z^*(m_0))} \frac{ \frac{ d z^*(0)}{d R}  \left( 1-e^{ -(\gamma_1+\gamma_2)  z^*(m_0)} \right)-\frac{ d z^*(m_0)}{d R} \left(1-e^{- (\gamma_1+\gamma_2)  z^*(0)}\right)   }{\left( e^{ (\gamma_1+\gamma_2)  z^*(m_0)}- 1 \right)^2}
 \nonumber
\end{align*}
\normalsize
Using that $\gamma_1+\gamma_2<0$ we have that
\begin{align} \label{Compa for analysis}
 & \frac{ d \Pr \left(M(t) > m_0, \min_{s \leq t} \{ Z(s) - Z^*(M(s)) \} >0 \right) }{dR(L)} >(<)0 \\
 & \quad \quad \quad \iff \quad 
 \frac{ d z^*(0)}{d R(K)}  \left( 1-e^{ -(\gamma_1+\gamma_2)  z^*(m_0)} \right)<(>)\frac{ d z^*(m_0)}{d R(L)} \left(1-e^{- (\gamma_1+\gamma_2)  z^*(0)}\right)   \nonumber
 \end{align}
 
 The next two claims provide comparative statics results that are needed to find changes in $z^*(0)$ and $z^*(m_0)$.

\begin{claim} \label{Claim comparative results}
The effects of $R$ are given by
\begin{align*}
\frac{d m}{d R} =
\begin{cases}
 \frac{1}{ (\pi(m_1)-L + c )g'(-z^*(m_1)) -\pi'(m_1)} & \text{ if } m=m_1\\
\frac{g(-z^*(m_1))}{R g'(-z^*(m_1))} e^{\int_{m_0}^{m_1} g'(-z^*(x)) dx} -\frac{1}{R g'(-z^*(m_1))} >0 & \text{ if } m=m_0
 \end{cases}
\end{align*}
and\footnote{Specifically and for future reference,
\[
\frac{d z^*(m_0)}{d R} =- \frac{1}{c g'(-z^*(m_0))} 
\]
}
\[
\frac{ d z^*(m)}{d R}  = 0 \text{ if } m >m_1 \quad \text{ and } \quad \frac{ d z^*(m)}{d R} < 0 \text{ if } m \in [m_0, m_1] \quad \text{ and } \quad \frac{ d z^*(m)}{d R} > 0 \text{ if } m < m_0
 \]
 while
 \begin{align*}
\frac{d z^*(m_1)}{d R}  =  \left(  \frac{\pi'(m_1)}{\pi(m_1)-L+c} \frac{g(-z^*(m_1))}{g'(-z^*(m_1))} -1 \right) \frac{d m_1}{d R}
\end{align*}

\end{claim}

\begin{proof}[Proof of Claim \ref{Claim comparative results}]

From the first two equations in \eqref{system to solve} we have for all $m \in (m_1,\overline{q}]$,
\begin{align*}
0 & =\frac{\partial z^*(m)}{\partial R} - \mathcal{H}_R(m) 
\end{align*}
where
\[
\mathcal{H}_R(m) \equiv \int_m^{\overline{q}}  h_R(x) \frac{ \partial z^*(x)}{\partial R} dx
\]
and 
\[
h_R(m) \equiv \frac{\pi'(m)}{\pi(m)-L+c}  \frac{g(-z^*(m))}{g'(-z^*(m))} \left( \frac{g'(-z^*(m))}{g(-z^*(m))}-\frac{g''(-z^*(m))}{g'(-z^*(m))} \right) <0
\]
Note that one solution is $\frac{\partial z^*(m)}{\partial R}=0$ for all $m >m_1$. Alternatively, assume that $\frac{\partial z^*(m)}{\partial R} \neq 0$
Note that $\mathcal{H}'_R(m) = -h_R(m) \frac{ \partial z^*(m)}{\partial R}$ so we get
\begin{align*}
0 & = \mathcal{H}'_R(m)+ h_R(m) \mathcal{H}_R(m) 
\end{align*}
with solution
\begin{align*}
 \mathcal{H}_R(m) & = A e^{ \int_m^{\overline{q}} h_R(x) dx}
\end{align*}
where we used the boundary condition $\mathcal{H}_R(\overline{q})=0$ and therefore
\[
 \frac{ \partial z^*(m)}{\partial R} = A e^{ \int_m^{\overline{q}} h_R(x) dx} 
 \]
Recall that $ \frac{ \partial z^*(m)}{\partial R}_{m=\overline{q}}=0$ so we must have that $A=0$, and therefore, the only solution is $\frac{\partial z^*(m)}{\partial R}=0$ for all $m>m_1$.

It follows then that the system is separable and the changes in $(z^*(m_1),m_1)$ are given by
\begin{align*}
\frac{d z^*(m_1)}{d R} &  =  \left(  \frac{\pi'(m_1)}{\pi(m_1)-L+c} \frac{g(-z^*(m_1))}{g'(-z^*(m_1))} -1 \right) \frac{d m_1}{d R}   \\
 \frac{d m_1}{d R}  & = \frac{1}{ (\pi(m_1)-L + c )g'(-z^*(m_1)) -\pi'(m_1)} 
\end{align*}
Now for $m \in (m_0,m_1)$ we have that
\begin{align*}
\frac{d z^*(m_1)}{d R}  + \frac{d m_1}{d R} & =\frac{\partial z^*(m)}{\partial R}  + g(-z^*(m_1))  \frac{d m_1}{d R}  - \int_{m}^{m_1} g'(-z^*(x)) \frac{\partial z^*(x)}{\partial R}  dx  & & m \in (m_0,m_1) \\
\frac{1}{(\pi(m_1)-L+c)} \frac{ g(-z^*(m_1))}{g'(-z^*(m_1))}  & = \int_{m}^{m_1} g'(-z^*(x)) \frac{\partial z^*(x)}{\partial R}  dx  - \frac{\partial z^*(m)}{\partial R}  & & m \in (m_0,m_1)
\end{align*}
Let $\mathcal{L}(m) \equiv \int_{m}^{m_1} g'(-z^*(x)) \frac{d z^*(x)}{d R}  dx $ so we have the differential equation
\begin{align*}
g'(-z^*(m)) \frac{1}{(\pi(m_1)-L+c)} \frac{ g(-z^*(m_1))}{g'(-z^*(m_1))}  & =g'(-z^*(m)) \mathcal{L}(m) + \mathcal{L}'(m)  & & m \in (m_0,m_1)
\end{align*}
and has the solution
\begin{align*}
\mathcal{L}(m) & =  \frac{1}{(\pi(m_1)-L+c)} \frac{ g(-z^*(m_1))}{g'(-z^*(m_1))} \left(   1- e^{\int_m^{m_1} g'(-z^*(x)) dx}  \right)   
\end{align*}
where we used the boundary condition $\mathcal{L}(m_1)=0$. Using the definition of $\mathcal{L}(m)$ we get
\[
\frac{1}{(\pi(m_1)-L+c)} \frac{ g(-z^*(m_1))}{g'(-z^*(m_1))} \left(   1- e^{\int_m^{m_1} g'(-z^*(x)) dx}  \right)  = \int_{m}^{m_1} g'(-z^*(x)) \frac{d z^*(x)}{d R}  dx
\]
and for all $m \in (m_0,m_1)$,
\[
\frac{\partial z^*(m)}{\partial R} = - \frac{1}{(\pi(m_1)-L+c)} \frac{ g(-z^*(m_1))}{g'(-z^*(m_1))}  e^{\int_m^{m_1} g'(-z^*(x)) dx}    <0 
\]

For the state $(z^*(m_0),m_0)$ we have
\begin{align*}
\frac{d z^*(m_1)}{d R}  +  \frac{d m_1}{d R} & =\frac{d z^*(m_0)}{d R}  -(g(-z^*(m_0))-1)   \frac{d m_0}{d R} +  g(-z^*(m_1)) \frac{d m_1}{d R}  - \int_{m_0}^{m_1} g'(-z^*(x)) \frac{ d z^*(x)}{d R} dx  \\
 \frac{d z^*(m_0)}{d R}  & =  - \frac{1}{c g'(-z^*(m_0))}  <0
\end{align*}
or using $\mathcal{M}(m)$, the previous results, and some algebra
\begin{align*}
R g'(-z^*(m_1))   \frac{d m_0}{d R}   & = g(-z^*(m_1)) e^{\int_{m_0}^{m_1} g'(-z^*(x)) dx} -1 \\
& > g(-z^*(m_1))  -1 >0  
\end{align*}

Now consider the states $m <m_0$ for which we have
\[
z^*(m_0)+m_0 = z^*(m)+m
\]
which implies
\begin{align*}
\frac{ d z^*(m)}{dR} & = \frac{ d z^*(m_0)}{dR}+\frac{ d m_0}{dR} \\
R \frac{ d z^*(m)}{dR} & =\frac{g(-z^*(m_1)) e^{\int_{m_0}^{m_1} g'(-z^*(x)) dx}-1}{ g'(-z^*(m_1)) }  - \frac{ g(-z^*(m_0))-1}{ g'(-z^*(m_0))} 
\end{align*}
Note that
\[
 \frac{ d \left( \frac{g(-z^*(m_1)) }{ g'(-z^*(m_1)) } \right)}{d z^*(m_1)}>0 
 \]
which implies that
\begin{align*}
R \frac{ d z^*(m)}{dR} & > \frac{g(-z^*(m_0)) e^{\int_{m_0}^{m_1} g'(-z^*(x)) dx}-1}{ g'(-z^*(m_0)) }  - \frac{ g(-z^*(m_0))-1}{ g'(-z^*(m_0))} \\
& > \frac{g(-z^*(m_0)) -1}{ g'(-z^*(m_0)) }  - \frac{ g(-z^*(m_0))-1}{ g'(-z^*(m_0))} =0
\end{align*}

\end{proof}

\medskip

\begin{claim} \label{Claim comparative results b}
The effects of $L$ are given by
\[
\frac{d m_0}{d L}   =- \frac{ 1}{g(-z^*(m_1))-1}  \left( \frac{d z^*(m_1)}{d L}  -(g(-z^*(m_1))-1) \frac{d m_1}{d L} \right) >0
\]
and
 \[
 \frac{ d z^*(m)}{d L}  < 0  \quad  \text{ if } m > m_1  \quad \text{ and } \quad  \frac{ d z^*(m)}{d L}  > 0  \quad  \text{ if } m < m_1
 \]
 and
 \[
 \frac{ d z^*(m_0)}{d L}  = 0
 \]
while
 $\frac{ d z^*(m_1)}{d L} $ and $ \frac{d m_1}{d L} $ are not determined.

\end{claim}

\begin{proof}[Proof of Claim \ref{Claim comparative results b}]

For the effects of $L$, let's focus first on $m >m_1$ where we have
\begin{align*}
0 & = \frac{ \partial z^*(m)}{\partial L} + \int_m^{\overline{q}}  \frac{\pi'(x)}{(\pi(x)-L+c)^2} \frac{g(-z^*(x))}{g'(-z^*(x))} dx - \mathcal{H}(m)  
% \\& \quad \quad - \int_m^{\overline{q}}  \frac{\pi'(x)}{\pi(x)-L+c} 
% \frac{g(-z^*(x))}{g'(-z^*(x))} \overset{(-)}{\overbrace{\left( \frac{g'(-z^*(x))}{g(-z^*(x))}-\frac{g''(-z^*(x))}{g'(-z^*(x))} \right)}} \frac{ d z^*(x)}{d L} dx
\end{align*}
where
\[
\mathcal{H}(m) \equiv \int_m^{\overline{q}}  h(x) \frac{ \partial z^*(x)}{\partial L} dx
\]
and 
\[
h(m) \equiv \frac{\pi'(m)}{\pi(m)-L+c}  \frac{g(-z^*(m))}{g'(-z^*(m))} \left( \frac{g'(-z^*(m))}{g(-z^*(m))}-\frac{g''(-z^*(m))}{g'(-z^*(m))} \right) <0
\]
Using the function $\mathcal{H}(m)$ we get
\begin{align*}
\mathcal{H}'(m) +h(m) \mathcal{H}(m) & =  h(m) \int_m^{\overline{q}}  \frac{\pi'(x)}{(\pi(x)-L+c)^2} \frac{g(-z^*(x))}{g'(-z^*(x))} dx  \\
\frac{d \left(\mathcal{H}(m) e^{-\int_m^{\overline{q}} h(x) dx } \right)}{d m} & =  \frac{d \left( e^{-\int_m^{\overline{q}} h(x) dx } \right) }{dm} \int_m^{\overline{q}}  \frac{\pi'(x)}{(\pi(x)-L+c)^2} \frac{g(-z^*(x))}{g'(-z^*(x))} dx  \\
\mathcal{H}(m) 
 & =      \int_m^{\overline{q}}  \left( 1 - e^{\int_m^y h(x) dx } \right) \frac{\pi'(y)}{(\pi(y)-L+c)^2} \frac{g(-z^*(y))}{g'(-z^*(y))} dy  
 \end{align*}
where we used that $\mathcal{H}(\overline{q})=0$. Using the definition of $ \mathcal{H}(m)$ we get
\begin{align*}
 - h(m) \frac{ \partial z^*(m)}{\partial L} & = \mathcal{H}'(m)   \\
\frac{ \partial z^*(m)}{\partial L} & = - \int_m^{\overline{q}}  e^{\int_m^y h(x) dx }  \frac{\pi'(y)}{(\pi(y)-L+c)^2} \frac{g(-z^*(y))}{g'(-z^*(y))} dy  <0
\end{align*}
for all $m >m_1$ and note that this implies that
\begin{align} \label{Intermediate for compa}
0 & < \mathcal{H}(m) < \int_m^{\overline{q}}  \frac{\pi'(x)}{(\pi(x)-L+c)^2} \frac{g(-z^*(x))}{g'(-z^*(x))} dx 
 \end{align}
for all $m >m_1$. In particular, we have that 
\begin{equation} \label{Important inequality}
\frac{ \partial z^*(m_1)}{\partial L}   <0
\end{equation}

Now consider the system to determine $(z^*(m_1),m_1)$ which is given by
\begin{align*}
\overline{q} & = z^*(m_1)+ m_1+ \int_{m_1}^{\overline{q}}  \frac{\pi'(x)}{\pi(x)-L+c} \frac{g(-z^*(x))}{g'(-z^*(x))} dx \nonumber\\
\pi(m_1)-L+ c & =  \frac{R}{g(-z^*(m_1))-1}  
\end{align*}
We first use the first line and \eqref{Intermediate for compa} to prove an intermediate result
\small
\begin{align*}
0 & =  \frac{ d z^*(m_1)}{d L} - \left(  \frac{\pi'(m_1)}{\pi(m_1)-L+c} \frac{g(-z^*(m_1))}{g'(-z^*(m_1))}  -1 \right)  \frac{ d m_1}{d L} + \int_{m_1}^{\overline{q}}  \frac{\pi'(x)}{(\pi(x)-L+c)^2} \frac{g(-z^*(x))}{g'(-z^*(x))} dx - \mathcal{H}(m_1) \\
0 & >  \frac{ d z^*(m_1)}{d L} - \left(  \frac{\pi'(m_1)}{\pi(m_1)-L+c} \frac{g(-z^*(m_1))}{g'(-z^*(m_1))}  -1 \right)  \frac{ d m_1}{d L} \\
0 & >\frac{g'(-z^*(m_1))}{g(-z^*(m_1) } \left(  \frac{ d z^*(m_1)}{d L} -(g(-z^*(m_1)-1) \frac{ d m_1}{d L} \right) - \left(  \frac{\pi'(m_1)}{\pi(m_1)-L+c}   -g'(-z^*(m_1)) \right)  \frac{ d m_1}{d L} 
\end{align*}
\normalsize
To get the point result we differentiate the above system to get
\small
\begin{align} \label{Useful system a}
 \frac{ d z^*(m_1)}{d L} & = \left(  \frac{\pi'(m_1)}{\pi(m_1)-L+c} \frac{g(-z^*(m_1))}{g'(-z^*(m_1))}-1 \right) \frac{ d m_1}{d L}  -  \int_{m_1}^{\overline{q}}   e^{\int_{m_1}^y h(x) dx } \frac{\pi'(y)}{(\pi(y)-L+c)^2} \frac{g(-z^*(y))}{g'(-z^*(y))} dy \\
  \frac{ d z^*(m_1)}{d L}  & =  \frac{\pi'(m_1) }{\pi(m_1)-L+c}  \frac{g(-z^*(m_1))-1}{g'(-z^*(m_1))} \frac{ d m_1}{d L}  - \frac{1}{g'(-z^*(m_1))} \frac{g(-z^*(m_1))-1}{\pi(m_1)-L+c} \nonumber
\end{align}
\normalsize
which imply that
\small
\begin{align} \label{for ineq a}
\frac{ \frac{\pi'(m_1)}{\pi(m_1)-L+c} - g'(-z^*(m_1)) }{g'(-z^*(m_1))}  \frac{ d m_1}{d L}  & =  \int_{m_1}^{\overline{q}}   e^{\int_{m_1}^y h(x) dx } \frac{\pi'(y)}{(\pi(y)-L+c)^2} \frac{g(-z^*(y))}{g'(-z^*(y))} dy  \\
& \quad \quad \quad      - \frac{g(-z^*(m_1))-1}{\pi(m_1)-L+c} \frac{}{g'(-z^*(m_1))}  \nonumber
\end{align}
\normalsize
and
\small
\begin{align} \label{for ineq b} 
\frac{ \frac{ d z^*(m_1)}{d L}-(g(-z^*(m_1)-1) \frac{ d m_1}{d L}}{g(-z^*(m_1))-1}  & = \int_{m_1}^{\overline{q}}   e^{\int_{m_1}^y h(x) dx } \frac{\pi'(y)}{(\pi(y)-L+c)^2} \frac{g(-z^*(y))}{g'(-z^*(y))} dy   \\
& \quad \quad   - \frac{g(-z^*(m_1))-1}{g'(-z^*(m_1))} \frac{1}{\pi(m_1)-L+c}   \frac{g(-z^*(m_1))}{g(-z^*(m_1))-1}  \nonumber
\end{align} 
\normalsize
or after some algebra
\small
\begin{align*}
 \left( \frac{ d z^*(m_1)}{d L} -(g(-z^*(m_1)-1) \frac{ d m_1}{d L} \right) & \frac{g'(-z^*(m_1))}{g(-z^*(m_1))} = \left(  \frac{\pi'(m_1)}{\pi(m_1)-L+c} -g'(-z^*(m_1)) \right)   \frac{ d m_1}{d L}  \\
&   - \frac{g'(-z^*(m_1))}{g(-z^*(m_1))}   \int_{m_1}^{\overline{q}}   e^{\int_{m_1}^y h(x) dx } \frac{\pi'(y)}{(\pi(y)-L+c)^2} \frac{g(-z^*(y))}{g'(-z^*(y))} dy \\
 \left( \frac{ d z^*(m_1)}{d L}  -(g(-z^*(m_1)-1) \frac{ d m_1}{d L} \right) & \frac{g'(-z^*(m_1))}{g(-z^*(m_1)-1)}  =  \left( \frac{\pi'(m_1)}{\pi(m_1)-L+c}-g'(-z^*(m_1)  \right)  \frac{ d m_1}{d L} -\frac{1}{\pi(m_1)-L+c}      
\end{align*}
\normalsize

Unfortunately, the effects on $m_1$ and $z^*(m_1)$ are not determined. 

Now consider $m \in (m_0,m_1)$ and \eqref{system to solve} gives 
\begin{equation} \label{Aux for compa}
\frac{d z^*(m_1)}{d L}  -(g(-z^*(m_1))-1)  \frac{d m_1}{d L}  =\frac{\partial z^*(m)}{\partial L} - \int_{m}^{m_1} g'(-z^*(x)) \frac{\partial z^*(x)}{\partial L} dx  
\end{equation}
Using \eqref{Important inequality} and taking limits as $m \uparrow m_1$ we get that
\[
\frac{d z^*(m_1)}{d L}  -(g(-z^*(m_1))-1)  \frac{d m_1}{d L} <0
\]
Let $\mathcal{M}(m) \equiv \int_{m}^{m_1} g'(-z^*(x)) \frac{\partial z^*(x)}{\partial L} dx$ to get
\begin{equation*} 
- g'(-z^*(m)) \left( \frac{d z^*(m_1)}{d L}  -(g(-z^*(m_1))-1)  \frac{d m_1}{d L}  \right) =\mathcal{M}'(m)  + g'(-z^*(m)) \mathcal{M}(m)  
\end{equation*}
with solution given by
\[
\mathcal{M}(m) = \left( e^{\int_m^{m_1} g'(-z^*(x)) dx} -1 \right)  \left( \frac{d z^*(m_1)}{d L}  -(g(-z^*(m_1))-1) \frac{d m_1}{d L} \right) <0
\]
and therefore for all $m \in (m_0,m_1)$
\[
\frac{ \partial z^*(m)}{\partial L}   = e^{\int_m^{m_1} g'(-z^*(x)) dx}   \left( \frac{d z^*(m_1)}{d L}  -(g(-z^*(m_1))-1) \frac{d m_1}{d L} \right) <
\]
which depends on what happens with the state $\left( z^*(m_1), m_1 \right)$. Now note that for all $m \in (m_0,m_1)$ 
\begin{align*}
\int_m^{m_1} g'(-z^*(x)) dx & = \int_m^{m_1} \frac{g'(-z^*(x))}{g(-z^*(x))-1} (g(-z^*(x))-1) dx \\
& = -\int_m^{m_1} \frac{\partial \ln \left( g(-z^*(x))-1 \right)}{\partial z^*(x)} \frac{\partial z^*(x)}{\partial x} dx \\
& = - \left(  \ln \left( g(-z^*(x))-1 \right) \right)_m^{m_1} \\
& =  \ln \left( \frac{ g(-z^*(m))-1}{g(-z^*(m_1))-1} \right) 
\end{align*}
and therefore
\[
e^{\int_m^{m_1} g'(-z^*(x)) dx} = e^{\ln \left( \frac{ g(-z^*(m))-1}{g(-z^*(m_1))-1} \right)} = \frac{ g(-z^*(m))-1}{g(-z^*(m_1))-1}
\]
which implies that
\begin{align*}
\frac{ \partial z^*(m)}{\partial L}   & = \frac{ g(-z^*(m))-1}{g(-z^*(m_1))-1}   \left( \frac{d z^*(m_1)}{d L}  -(g(-z^*(m_1))-1) \frac{d m_1}{d L} \right) <0 \\
\mathcal{M}(m) & = \frac{ g(-z^*(m))-g(-z^*(m_1))}{g(-z^*(m_1))-1}  \left( \frac{d z^*(m_1)}{d L}  -(g(-z^*(m_1))-1) \frac{d m_1}{d L} \right)
\end{align*}

Consider now the state $\left( z^*(m_0), m_0 \right)$ and noting that $\frac{d z^*(m_0)}{d L} =0$ we have
\begin{align*} 
(g(-z^*(m_0))-1)\frac{d m_0}{d L}   & =- \int_{m_0}^{m_1} g'(-z^*(x)) \frac{\partial z^*(x)}{\partial L} dx - \left( \frac{d z^*(m_1)}{d L}  -(g(-z^*(m_1))-1)  \frac{d m_1}{d L} \right) \\
\frac{R}{c} \frac{d m_0}{d L} & =- \mathcal{M}(m_0)  - \left( \frac{d z^*(m_1)}{d L}  -(g(-z^*(m_1))-1)  \frac{d m_1}{d L} \right) \\
\frac{R}{c} \frac{d m_0}{d L}  & =- \frac{ g(-z^*(m_0))-1}{g(-z^*(m_1))-1}  \left( \frac{d z^*(m_1)}{d L}  -(g(-z^*(m_1))-1) \frac{d m_1}{d L} \right) > 0
\end{align*}
 
It is easy to see that for $m<m_0$ we have that
\[
\frac{dm_0}{d L}  = \frac{dz^*(m)}{d L}
\]

\end{proof}

From Claim \ref{Claim comparative results}, we have that 
\begin{align*}
\frac{d m_0}{d R} & = \frac{g(-z^*(m_1))}{R g'(-z^*(m_1))} e^{\int_{m_0}^{m_1} g'(-z^*(x)) dx} -\frac{1}{R g'(-z^*(m_1))} >0\\
\frac{d z^*(m_0)}{d R} &=- \frac{1}{c g'(-z^*(m_0))} <0
\end{align*}
and from Claim \ref{Claim comparative results b}, we have
\begin{align*}
\frac{d m_0}{d R} & =-e^{\int_{m_0}^{m_1} g'(-z^*(x)) dx} \frac{c}{R} \left( \frac{d z^*(m_1)}{d L}  -(g(-z^*(m_1))-1)  \frac{d m_1}{d L}  \right )    \\
\frac{d z^*(m_0)}{d R} &=0
\end{align*}
and using $\frac{d m_0}{d R} +\frac{d z^*(m_0)}{d R}=\frac{d z^*(0)}{d R}$ we get the result by replacing in \eqref{Compa for analysis}.

Finally, from Lemma \ref{Lemma comparative results c} we have that $\frac{d z^*(m_0)}{d \overline{q}}=0< \frac{d z^*(0)}{d \overline{q}}$ which gives the result.

\qed

%\end{proof}

\newpage

\section{The running Maximum} \label{app}

This section reproduces some of the results in \cite{book:KS1991}  and \cite{book:privault2022}, in particular Chapter 10 from \cite{book:privault2022}

\subsection{Definition} As usual we define the Wiener process 
\[
B_t \sim \mathbb{N}(0,t)
\]
and an undrifted scaled brownian motion
\[
W_t \sim \mathbb{N}(0,\sigma^2 t)
\]
We let the running maximum be given by
\[
B^*_t \equiv \max_{s \leq t} \left\{ \: B(s) \: \right\}
\]
Analogously
\[
W^*_t \equiv \max_{s \leq t} \left\{ \: W(s) \: \right\}
\]
Note that $W^*_t = \sigma B^*_t$

We define analogously a brownian motion with drift $\mu>0$ and variance $\sigma$ as
\[
X_t = \mu t + W^*_t 
\]
and we let its running maximum be
\[
X^*_t \equiv \max_{s \leq t} \left\{ \: X(s) \: \right\}
\]

\subsection{Distribution(s)} We focus first on a standard brownian motion, by looking at the undrifted scaled brownian motion $W^*_t $ that distributes like $\mathbb{N} (0, \sigma t)$. For any $a > 0$, we want to calculate the probability function $\Pr(W^*_t \leq a)$ and the density function $\varphi{W}^*_t(x)$ of $W^*_t$ so that
\[
1-\Pr(W^*_t \leq a) = \int_a^{\infty} \varphi_{W^*_t}(x) dx
\]
We are going to apply the Reflection principle. Let
\[
\tau_a \equiv \inf \left\{ t \geq 0, W_t \geq a \right\} =\inf \left\{ t \geq 0, B_t \geq \frac{a}{\sigma} \right\} 
\]
be the first time at which $W_t =a$ and let's define the reflected process
\[
W^r_t \equiv 
\begin{cases}
    W_t & \text{ if } t \leq \tau_a \\
    2a- W_t & \text{ if } t > \tau_a
\end{cases}
\]
Note that 
\begin{align*}
\Pr(W^*_t \geq a) & = \Pr(W^*_t \geq a, W_t < a) + \Pr(W^*_t \geq a, W_t \geq a) \\
& = \Pr(W^*_t \geq a, 2a- W^r_t < a) + \Pr(W^*_t \geq a \: | \: W_t \geq a) \times \Pr( W_t \geq a) \\
& = \Pr(W^*_t \geq a, W^r_t > a) +  \Pr( W_t \geq a) \\
& = \Pr( W^r_t > a) +  \Pr( W_t \geq a)  = 2 \Pr( W_t \geq a) 
\end{align*}
and therefore
\begin{align*}
 \int_a^{\infty} \varphi_{W^*_t}(x) dx & = 2 \Pr \left( W_t \geq a \right) \\
   & = 2 \mathbb{N}_{(0, t \sigma^2 )} (a)  = 2 \int_{\frac{a}{\sqrt{t \sigma^2}}} ^ \infty \phi(x) dx 
\end{align*}
and therefore
\begin{equation} \label{pdfMaxWiener}
\varphi_{W^*_t}(a) =  \frac{2}{\sqrt{2 \pi \sigma t}} e^{-\frac{a^2}{2 \sigma^2 t}} dx
\end{equation}
Analogously
\begin{align*}
    \Pr \left( W^*_t \geq a , W_t< b \right) & = \Pr \left( W^*_t \geq a , 2a- W^*_t< b \right) \\
    & = \Pr \left( W^*_t \geq a , W^*_t > 2a- b \right) \\
    & = \Pr \left( W^*_t \geq a \: | \: W^*_t > 2a- b \right) \\
    & = \Pr \left( W^*_t > 2a- b\right) = \Pr \left( W_t > 2a- b\right) 
\end{align*}
where the third line follows by the fact that $a>b$ and the last line follows by symmetry around the new centered $a$. It follows then that
\begin{align*}
\int_a^\infty \left( \int_{-\infty}^b \varphi_{W^*_t,W_t}(x,y) dy \right) dx  & = \Pr \left( W_t > 2a- b\right) \\
& = \mathbb{N}_{(0, t \sigma^2 )} (2a-b)  = \int_{-\infty}^{\frac{2a-b}{\sqrt{t \sigma^2}}} \phi(x) dx
\end{align*}
and therefore
\[
\varphi_{W^*_t,W_t}(a,b) =  \frac{2a-b}{t \sigma^2} \frac{2}{\sqrt{2 \pi t \sigma^2} } e^{\frac{(2a-b)^2}{2 t \sigma^2}} =  \frac{2a-b}{t \sigma^2} \varphi_{W^*_t}(2a-b)
\]

\underline{We focus now on a a brownian motion with drift}
\[
X_t = \mu t + W_t \quad \quad W_t \sim \mathbbm{N} (0, \sigma^2 t)
\]
and its running maximum
\[
X^*_t \equiv \max_{s \leq t} \left\{ X_t \right\}
\]
This case is more involved since symmetry occurs around $a-\mu \tau_a$ and the application of the reflection principle is not immediate. Nevertheless, it follows by a change of measure  ( see Girsanov's Theorem, \cite{book:KS1991}, page 302). A simpler statement is directly applicable to our problem:

\begin{proposition}[Girsanov] \label{prop:Girsanov}
Let $(\mathbbm{P}, \Omega, \mathcal{F})$ be a probability space, $\mathbb{F} = \left\{ \mathcal{F})_t \right\}_{t \geq 0}$ a complete filtration, and  $W_t \sim \mathbbm{N}(0,\sigma^2 t)$ an undrfited brownian motion adapted to $\mathbb{F}$. Let $X_t =  \mu t + W_t$ be a brownian motion with drift $\mu \in \mathbb{R}$ adapted to $\mathbb{F}$. Under the probability measure $\mathbbm{Q}$ defined by the Radon-Nykodin density
\[
\frac{d \mathbbm{Q}}{d\mathbbm{P}} \equiv e^{-\frac{\mu}{\sigma^2} W_t - \left( \frac{\mu}{\sigma}\right)^2 \frac{t}{2}}
\]
the random variable $X_t$ has centered gaussian distribution $\mathbbm{N}(0,\sigma t)$.  
\end{proposition}
\begin{proof}[Proof of Proposition \ref{prop:Girsanov}]
    Recall the definition of moment generating function of a Normal random variable $X \sim \mathbbm{N}(\hat{\mu},\hat{\sigma})$
    \[
    M(z) = \mathbbm{E} [e^{z X}] = e^{z \hat{\mu} + z^2 \frac{\hat{\sigma}^2}{2}}
    \]
    In particular, in our case we have
    \[
    M_{W_t}(z) = \mathbbm{E} [e^{z W}] = e^{z^2 \sigma ^2 \frac{ t}{2}}
    \]
    Note that
        \begin{align*}
        \mathbbm{E}_{\mathbbm{Q}} (e^{z X_t})   & =         \mathbbm{E}_{\mathbbm{P}} \left(\frac{d \mathbbm{Q}}{d\mathbbm{P}} e^{z X_t} \right) \\ 
        & =         \mathbbm{E}_{\mathbbm{P}} \left(e^{-\frac{\mu}{\sigma^2} W_t - \left( \frac{\mu}{\sigma}\right)^2 \frac{t}{2}} e^{z \mu t + z W_t} \right) \\ 
		& =   \mathbbm{E}_{\mathbbm{P}} \left(e^{\left( z  -\frac{\mu}{\sigma^2} \right) W_t } \right) e^{\left( 2z \mu  - \left( \frac{\mu}{\sigma}\right)^2 \right)\frac{t}{2} }  \\
		& =  e^{\left( z  -\frac{\mu}{\sigma^2} \right)^2 \frac{\sigma^2 t}{2}  }  e^{\left( 2z \mu  - \left( \frac{\mu}{\sigma}\right)^2 \right)\frac{t}{2} }  =  e^{\left( \left( z \sigma  -\frac{\mu}{\sigma} \right)^2+ 2z \mu  - \left( \frac{\mu}{\sigma}\right)^2 \right)\frac{t}{2} } =    e^{ \left( z \sigma \right)^2 \frac{t}{2} }   
    \end{align*}
    so $X_t$ is a brownian motion without drift measurable in the space  $(\mathbbm{Q}, \Omega, \mathcal{F})$
\end{proof}
It remains to calculate the cumulative distribution function of $X_t$ in the proper space $(\mathbbm{P}, \Omega, \mathcal{F})$. The following is an adaptation of Proposition 10.3 in \cite{book:privault2022} to our particular case.

\begin{proposition}\label{Prop:Dist of RM}
Let $X_t$ be a brownian motion with drift $\mu$ and variance $\sigma$ and let $X^*_t$ be its running maximum. The joint density function for any pair $a > \max\{b,0\}$ is given by
\begin{align*}
\varphi_{X^*_t,X_t}(a,b) & = e^{\frac{\mu}{\sigma^2} b -  \left( \frac{\mu}{\sigma} \right)^2 \frac{t}{2}}  \frac{2a-b}{\sigma^2 t} \frac{2}{\sqrt{2 \pi \sigma^2 t}} e^{-\frac{(2a-b)^2}{2\sigma^2t}} \\
& =     e^{\frac{\mu}{\sigma^2} \left(  \frac{2b-\mu t}{2} \right) }  \varphi_{W^*_t,W_t}(a,b) 
% \\
% &=  2 \frac{2a-b}{t} \frac{1}{\sqrt{2 \pi t}}  e^{  - \frac{(b- \mu t ) ^2 + 2 a (a-b) }{2t}} \\
% & =  2 \frac{2a-b}{t} \frac{1}{\sqrt{2 \pi t}}  e^{  - \frac{(a+\mu t - b)^2 +a (a-2 \mu t) }{2t}}
\end{align*}
\end{proposition}
\begin{proof}[Proof of Proposition \ref{Prop:Dist of RM}]
    Let $W_t$ be a $(0,\sigma^2t)$ brownian motion and note that we are interested in $X_t= \mu t +W_t$. For any $a \geq \max \{b,0\}$, and using the same notation as in Proposition \ref{prop:Girsanov} we have
\begin{align*}
\Pr(X^*_t \leq a, X_t \leq b ) & = \mathbbm{E}_{\mathbbm{P}} \left[ \mathbbm{1} (X^*_t \leq a, X_t \leq b) \right] \\
& = \mathbbm{E}_{\mathbbm{Q}} \left[  \frac{d\mathbbm{P}}{d \mathbbm{Q}} \mathbbm{1} (X^*_t \leq a, X_t \leq b) \right] \\
& = \mathbbm{E}_{\mathbbm{Q}} \left[  e^{\frac{\mu}{\sigma^2} W_t + \left( \frac{\mu}{\sigma}\right)^2 \frac{t}{2}} \mathbbm{1} (X^*_t \leq a, X_t \leq b) \right] \\
& = \mathbbm{E}_{\mathbbm{Q}} \left[  e^{\frac{\mu}{\sigma^2} (X_t - \mu t) + \left( \frac{\mu}{\sigma}\right)^2 \frac{t}{2}} \mathbbm{1} (X^*_t \leq a, X_t \leq b) \right] \\
& = \mathbbm{E}_{\mathbbm{Q}} \left[  e^{\frac{\mu}{\sigma^2} X_t - \left( \frac{\mu}{\sigma}\right)^2 \frac{t}{2}} \mathbbm{1} (X^*_t \leq a, X_t \leq b) \right] \\
& = \int_0^a \left[ \int_{-\infty}^b  \mathbbm{1}[y < x]e^{\frac{\mu}{\sigma^2} y -  \left( \frac{\mu}{\sigma} \right)^2 \frac{t}{2}}  \frac{2x-y}{\sigma^2 t} \frac{2}{\sqrt{2 \pi \sigma^2 t}} e^{-\frac{(2x-y)^2}{2\sigma^2t}} dy \right]  dx 
\end{align*}
Taking the derivative with respect to both $a$ and $b$ we get the result.
\end{proof}

It is immediate to see that for any $a > \max\{b,0\}$
\begin{align} \label{PDFMaxWiener*1}
    \varphi_{X^*_t\: | \: X_t=b}(a) & = 
    2   \frac{2a-b}{\sigma^2 t} e^{-\frac{2a(a-b)}{\sigma^2t}}  
\end{align}
and integrating over $b$ we get
\begin{align}  \label{PDFMaxWiener*2}
\varphi_{X^*_t}(a) & = \int_{-\infty}^ a \varphi_{X^*_t,X_t}(a,b) db \\
& = \int_{-\infty}^ a  \frac{2a-b}{\sigma^2 t} \frac{2}{\sqrt{2 \pi \sigma^2 t}} e^{-\frac{(2a-b)^2-2 b \mu t+(\mu t) ^2 }{2\sigma^2t}} db \nonumber \\
& = e^{\frac{2a \mu }{\sigma^2}} \int_{-\infty}^ a  \frac{2a-b}{\sigma^2 t} \frac{2}{\sqrt{2 \pi \sigma^2 t}} e^{-\frac{(2a+ \mu t-b)^2}{2\sigma^2t}} db \nonumber \\
& = - 2 e^{\frac{2a \mu }{\sigma^2}} \int_{\infty}^{\frac{a+ \mu t}{\sqrt{\sigma^2t}}}  \frac{x \sqrt{\sigma^2t} - \mu t}{\sigma^2 t} \phi(x) dx \nonumber \\
& =  \frac{2}{\sqrt{\sigma^2t}} e^{\frac{2a \mu }{\sigma^2}} \int_{\infty}^{\frac{a+ \mu t}{\sqrt{\sigma^2t}}}  \phi'(x) dx 
+  2 e^{\frac{2a \mu }{\sigma^2}} \frac{ \mu}{\sigma^2} \int_{\infty}^{\frac{a+ \mu t}{\sqrt{\sigma^2t}}}   \phi(x) dx \nonumber \\
& =  \frac{2}{\sqrt{\sigma^2t}}  \phi \left(\frac{a - \mu t}{\sqrt{\sigma^2t}} \right)  
+   \frac{ 2 \mu}{\sigma^2}e^{\frac{2a \mu }{\sigma^2}}  \Phi \left( \frac{a+ \mu t}{\sqrt{\sigma^2t}} \right) \nonumber
\end{align}
Finally, we get that
\begin{align} \label{PDFMaxWiener*3}
\Pr(X^*_t \leq m) & = \int_0^m  \varphi_{X^*_t}(a) da \\
& =   \frac{2}{\sqrt{\sigma^2 t}} \int_0^m \phi \left( {\frac{a- \mu t}{\sqrt{\sigma^2 t}}} \right)  da -2 \frac{\mu  } {\sigma^2} \int_0^m   e^{2 \frac{\mu a}{\sigma^2} }  \Phi\left(-\frac{a+\mu t}{\sqrt{\sigma^2 t}} \right) da \nonumber \\
& =  2 \left( \Phi \left(  \frac{m- \mu t}{\sqrt{\sigma^2 t}} \right) - \Phi \left(  - \frac{ \mu t}{\sqrt{\sigma^2 t}} \right) \right)  -  \left[ \int_0^m   \frac{2 \mu  } {\sigma^2} e^{2 \frac{\mu a}{\sigma^2} }  \Phi\left(-\frac{a+\mu t}{\sqrt{\sigma^2 t}} \right) da \right] \nonumber \\
& =  2  \Phi \left(  \frac{m- \mu t}{\sqrt{\sigma^2 t}} \right) - \Phi \left(  - \frac{ \mu t}{\sqrt{\sigma^2 t}} \right)   - 
\left[ 
   e^{2 \frac{\mu m}{\sigma^2} }  \Phi\left(-\frac{m+\mu t}{\sqrt{\sigma^2 t}} \right)  
 +  \frac{1}{\sqrt{\sigma^2 t}} \int_0^m   \phi\left(\frac{a-\mu t}{\sqrt{\sigma^2 t}} \right) da \right] \nonumber \\
& =   \Phi \left(  \frac{m- \mu t}{\sqrt{\sigma^2 t}} \right)   -    e^{2 \frac{\mu m}{\sigma^2} }  \Phi\left(-\frac{m+\mu t}{\sqrt{\sigma^2 t}} \right)  \nonumber
\end{align}

The following proposition summarizes the issues arising from the stochastic process $\sup_{s \leq t} X(s) = X^*_t$ when $X(s)$ is a $(\mu,\sigma)$ brownian motion.

\begin{proposition} \label{Prop: sup issues}
Let $X(s)$ be a $(\mu,\sigma)$ brownian motion adapted to the filtered space $(\mathbb{P}, \Omega, \mathcal{F}, \mathbb{F})$, then the supremum of it has nil quadratic variation and
\[
\lim_{dt \rightarrow 0} \frac{dX^*_t}{dt} \rightarrow \infty
\]
\end{proposition}

\begin{proof}[Proof of Proposition \ref{Prop: sup issues}]
For any for any partition $\mathcal{P} = \left\{\underline{t}= t_0<t_1<....t_N=\underline{t} \right\}$ and $|\mathcal{P} | = \max_k (t_{k+1}-t_k)$
\begin{align*}
 \langle X^*_t \rangle_{\underline{t}}^{\underline{t} } & = \lim_{|\mathcal{P} |  \rightarrow 0} \sum_0^{N-1} (X^*(t_{j+1})-X^*(t_{j}))^2  \\
  & \leq \lim_{|\mathcal{P} |  \rightarrow 0} \max_k \{ (X^*(t_{k+1})-X^*(t_{k}) \}\sum_0^{N-1} (X^*(t_{j+1})-X^*(t_{j}))  \\
    & \leq \lim_{|\mathcal{P} |  \rightarrow 0} \max_k \{ (X^*(t_{k+1})-X^*(t_{k}) \}  (X^*(\underline{t} )-X^*(\underline{t} ))  
\end{align*}
and since $\lim_{\eta \rightarrow \infty } \Pr (X^*(\underline{t} )-X^*(\underline{t} ) > \eta) \rightarrow 0 $ and 
$\max_k \{ (X^*(t_{k+1})-X^*(t_{k}) \} \rightarrow 0$ when $\lim_{|\mathcal{P} |  \rightarrow 0}$ we get the first part of the result.

Note now that
\[
\Pr(X^*(dt) > m dt) = 1- \Pr(X^*(dt) \leq m dt)
\]
and using \eqref{PDFMaxWiener*3} we get
\begin{align*}
\Pr(X^*(dt) > m dt) &= 1 - \Phi \left(  \frac{m dt- \mu dt}{\sqrt{\sigma^2 dt}} \right)   +    e^{2 \frac{\mu m dt}{\sigma^2} }  \Phi\left(-\frac{m dt+\mu dt}{\sqrt{\sigma^2 dt}} \right) \\
& = 1 - \Phi \left(  \frac{m- \mu}{\sqrt{\sigma^2}} \sqrt{ dt} \right)   +    e^{2 \frac{\mu m dt}{\sigma^2} }  \Phi\left(-\frac{m+\mu }{\sqrt{\sigma^2}}  \sqrt{dt} \right)  
\end{align*}
and taking limits as $dt \rightarrow 0$ we get $\Pr(X^*(dt) > m dt) = 1$. Moreover, it is easy to see that 
\begin{align*}
\Pr(X^*(dt) > m \sqrt{dt}) &= 1 - \Phi \left(  \frac{m \sqrt{dt}- \mu dt}{\sqrt{\sigma^2 dt}} \right)   +    e^{2 \frac{\mu m \sqrt{dt}}{\sigma^2} }  \Phi\left(-\frac{m \sqrt{dt}+\mu dt}{\sqrt{\sigma^2 dt}} \right) \\
& \rightarrow 1 - \Phi \left(  \frac{m}{\sqrt{\sigma^2}} \right)   +     \Phi\left(-\frac{m }{\sqrt{\sigma^2}} \right) = 2  \Phi\left(-\frac{m }{\sqrt{\sigma^2}} \right)  
\end{align*}
and since
\[
\mathbb{E} \left[ \frac{X^*(dt)}{dt} | X^*(0)=0 \right] \geq m \Pr \left[ \frac{X^*(dt)}{dt} >m | X^*(0)=0 \right] 
\]
we get that $\frac{dX^*_t}{dt}$ explodes when $X_t=X^*_t$ if $dt \rightarrow 0$.
\end{proof}

\end{document}